\documentclass[letterpaper,11pt]{article}

\usepackage[runin]{abstract}
\usepackage{geometry}
\usepackage{titlesec}
\usepackage{graphicx}
\usepackage{amsmath}
\usepackage{amsthm}
\usepackage{amsfonts}
\usepackage{amssymb}
\usepackage{empheq}
\usepackage{algorithmic}
\usepackage{algorithm}
\usepackage[authoryear]{natbib}
\usepackage{url}

\usepackage[font=sf]{caption}

\DeclareMathOperator*{\argmin}{arg\,min}

\usepackage{mathtools}

\DeclarePairedDelimiter\floor{\lfloor}{\rfloor}

\titleformat*{\section}{\Large\bfseries}
\titleformat*{\subsection}{\large\sc}
\titleformat*{\subsubsection}{\itshape}

\usepackage{epigraph}

\usepackage{etoolbox}

\makeatletter
\patchcmd{\epigraph}{\@epitext{#1}}{\itshape\@epitext{#1}}{}{}
\makeatother

\begin{document}

\title{{\bf On algorithmically boosting fixed-point computations}}

\author{\large{Ioannis Avramopoulos}\thanks{RelationalAI, Inc.} \and \large{Nikolaos Vasiloglou}\footnotemark[1]}

\date{}

\maketitle

\thispagestyle{empty} 

\newtheorem{definition}{Definition}
\newtheorem{proposition}{Proposition}
\newtheorem{theorem}{Theorem}
\newtheorem*{theorem*}{Theorem}
\newtheorem{corollary}{Corollary}
\newtheorem{lemma}{Lemma}
\newtheorem{axiom}{Axiom}
\newtheorem{thesis}{Thesis}

\vspace*{-0.2truecm}

\begin{abstract}
The main topic of this paper are algorithms for computing Nash equilibria. We cast our particular methods as instances of a general algorithmic abstraction, namely, a method we call {\em algorithmic boosting}, which is also relevant to other fixed-point computation problems. Algorithmic boosting is the principle of computing fixed points by taking (long-run) averages of iterated maps and it is a generalization of exponentiation. We first define our method in the setting of nonlinear maps. Secondly, we restrict attention to convergent linear maps (for computing dominant eigenvectors, for example, in the PageRank algorithm) and show that our algorithmic boosting method can set in motion {\em exponential speedups in the convergence rate}. Thirdly, we show that algorithmic boosting can convert a (weak) non-convergent iterator to a (strong) convergent one. We also consider a {\em variational approach} to algorithmic boosting providing tools to convert a non-convergent continuous flow to a convergent one. Then, by embedding the construction of averages in the design of the iterated map, we constructively prove the existence of Nash equilibria (and, therefore, Brouwer fixed points). We then discuss implementations of averaging and exponentiation, an important matter even for the scalar case. We finally discuss a relationship between dominant (PageRank) eigenvectors and Nash equilibria.
\end{abstract}

\section{Introduction}

From a purely algorithmic perspective, the main idea explored in this paper is that {\em exponentiation is a form of averaging} by which we treat {\em exponentiation} (for example, raising a vector to an exponent) and {\em averaging} (for example, averaging the orbit of a dynamical system) under the same analytical footing. We capture this intuition in a new rigorous definition of the exponential function. This definition is one of our primary conceptual contributions. But at a more elementary technical level, this paper is an inquiry into the computational foundations of fixed point theory. Various fixed point theorems (such as the Brouwer fixed-point theorem and the Knaster-Tarski theorem) are non-constructive and our ultimate goal is to develop algorithms by which this gap can be bridged. We first look at constructive fixed point theorems such as the Perron-Frobenius theorem and the minimax theorem, which are, in fact, related by a theorem of \cite{Blackwell1}. We then devise a {\bf constructive existence proof of Nash equilibrium} (and, therefore, of Brouwer fixed points). Toward obtaining such constructive proof, we leverage a powerful algorithmic abstraction whose development was driven by a curiosity question, namely, what, if any, role {\em exponentiation} can play in such a computational theory of fixed points. Exponentiation plays a key role in computational learning theory \citep{Littlestone, COLTpaper, Boosting}. Learning theory is grounded on a different analytical foundation than fixed point theory, nevertheless, the approaches are related as it is typically a small step to convert an (online) learning algorithm to a {\em fixed-point iterator} (for example, by simply assuming the online adversary is nature).

\subsection{Our main idea in simple terms}

Our curiosity had many fruits to bear: Our thesis is that {\em applying exponentiation to algorithmic problems is a powerful pursuit} and our line of discourse can be appositely framed by asking the question: {\bf How can we exponentiate entire algorithms?} Our idea to answer this question is the observation that, if $f( \cdot )$ is a self-map on the set of real numbers ($\mathbb{R}$), by redefining the exponential of $f$ as
\begin{align*}
\exp(f(x)) = \sum_{k=0}^{\infty} \frac{1}{k!} f^k(x),
\end{align*}
where $f^k(x) = f(f( \cdots f(x)))$, we can meaningfully think of exponentiation as an averaging process (in signal processing terms, as a {\em filter}). Normalizing the previous expression by dividing with $\exp(1)$, we obtain an exponential operator as a convex combination of the powers of $f$ that, as we will prove shortly, inherits the fixed points of $f$. That is, if $x^*$ is a fixed point $f$, in that $f(x^*)  = x^*$, then
\begin{align*}
\frac{1}{\exp(1)} \exp(f(x^*)) = x^*.
\end{align*}
To apply this (averaging) idea to algorithms, we need to think of a function as an algorithm and the analogy is immediately apparent if we specifically look at algorithms that compute fixed points: Various problems in computer science can be cast as instances of computing a fixed point of a map. Given a set of vectors, say $X$, and a self-map $f(\cdot)$ on $X$, a {\em fixed point} of $f$ is an $x^* \in X$ such that $f(x^*) = x^*$. In this paper, we focus on one particular method of computing a fixed point of a map, namely, the iterative application of either the map itself or some map that is naturally related to it. We further focus on two particular problem domains, namely, {\em linear algebra} and {\em game theory}.

\subsection{Exponentiation in linear algebra}

If $X$ is Euclidean space and $f$ is a linear map, then the action of $f$ on $X$ can be represented by a square matrix, say $A$. That is, $f(x) = Ax$. In this case, a fixed point of $f$ is a vector $x^* \in X$ such that $Ax^* = x^*$. That is, $x^*$ is an eigenvector of $A$ corresponding to eigenvalue $1$, if it exists. A prominent example of an algorithm that can be cast as a problem of computing an eigenvector corresponding to the eigenvalue $1$ is PageRank \citep{PageRank}. In fact, our algorithmic boosting theory began as an effort to apply exponentiation to the PageRank algorithm.

\subsubsection{Trying to exponentiate PageRank}

One approach to apply exponentiation to the PageRank algorithm is in the paradigm of the {\em multiplicative weights update method} \citep{AHK}, which is a very successful paradigm of designing algorithms that manifests in various disciplines of theoretical computer science. For example, in theoretical machine learning, exponentiated multiplicative updates manifest in {\em boosting,} which refers to a method of producing an accurate prediction rule by combining inaccurate prediction rules \citep{Boosting}. A well-established boosting algorithm is AdaBoost \citep{FreundSchapire1}. Related to AdaBoost is the Hedge algorithm for playing a mathematical game \citep{FreundSchapire2}. At the heart of AdaBoost and Hedge lies the weighted majority algorithm \citep{Littlestone} (see also \citep{COLTpaper}), which is also based on exponentiation. It was natural to ponder whether PageRank bears the structure of one of these problems. In fact, if we could spot a gradient we could just exponentiate that. The answer we give in this paper is unlikely to be unique but it motivated basic results in algorithmic exponentiation.

\subsubsection{Exponentiating PageRank by exponentiating the Google matrix}

{\em The critical point in the development of the theory we lay out in this paper was in realizing that PageRank bears the structure of a linear fixed point problem as it was then a no-brainer to exponentiate the matrix operator itself.} An advantage of this approach is that our method applies with minor or no modifications to other computational link-analysis problems such that {\em spectral rankings} \citep{Vigna} and {\em community detection} \citep{Newman}. In fact, the {\em exponentiated power method} and its {\em truncation} we propose in this paper are new ideas in numerical linear algebra \citep{Golub} and it is perhaps surprising that they admit a very simple analysis to theoretically ground the acceleration benefits. Our fully exponentially powered method consists in applying power iterations using the {\em matrix exponential} which is defined by the power series
\begin{align*}
\exp(A) = \sum_{k=0}^{\infty} \frac{1}{k!} A^k
\end{align*}
directly generalizing the previous definition of the exponential function. Our exponentiated power method boosts the convergence rate of the simple power method by an exponential factor (while retaining the character of the simple power method, simply by replacing the matrix operator being iterated with an exponentiated version thereof). This result may appear at first to only be principally theoretical as it is, in general, not possible (or even practical) to compute the matrix exponential exactly. However, we show that using only $m > 0$ terms in the power series we obtain a similar convergence rate (by a practical iteration algorithm), thus, reaping the acceleration benefits while obviating the need to exactly compute the exponential matrix. Owing to their simplicity (retaining the philosophy of the simple power iteration method, which works well in practice especially as the size of the matrix being iterated is large), our exponentially powered acceleration methods based on averaging the powers of a corresponding matrix operator can be combined with standard techniques for accelerating the convergence of sequences (and, for example, the computation of the PageRank vector by the power method) such as the {\em vector epsilon algorithm} (see \citep{Brezinski}). The epsilon algorithm and related Pad\'e approximation theory are also discussed slightly more thoroughly later in this section (after discussing mathematical games first).

\subsection{A generalization of the exponential matrix}

Once such a powerful idea of computing fixed points by averaging powers of linear operators has been set into place it is natural to try to apply it to fixed-point problems that we don't have efficient algorithms for, for example, Nash equilibrium problems in game theory (where the operators acting on vectors are nonlinear). We, thus, generalize the matrix exponential in the following fashion:

\subsubsection{Fully exponentially powered boosting}

\begin{definition}
Let $X$ be a convex set of vectors and $f( \cdot )$ a self-map on $X$. Given $\alpha > 0$, we define the fully exponentially powered self-map $\exp(\alpha f( \cdot ))$ as
\begin{align*}
\exp(\alpha f(x)) = \frac{1}{\exp(\alpha)} \sum_{k = 0}^{\infty} \frac{1}{k!} \alpha^k f^k(x)
\end{align*}
where
\begin{align*}
\exp(\alpha) = \sum_{k=0}^{\infty} \frac{1}{k!} \alpha^k
\end{align*}
and
\begin{align*}
f^0(x) = x \quad f^1(x) = f(x) \quad f^2(x) = f( f(x) ) \quad \cdots.
\end{align*}
We call $\alpha$ the learning rate of the fully exponentially powered map.
\end{definition}

\begin{lemma}
Let $X$ be a convex set of vectors and $f( \cdot )$ a self-map on $X$. Then, if $x^*$ is a fixed point of $f$, that is, if $f(x^*) = x^*$, it is also a fixed point of $\exp(\alpha f(\cdot))$.
\end{lemma}

\begin{proof}
Since $x^*$ is a fixed point of $f$, we obtain that
\begin{align*}
\exp(\alpha f(x^*)) = \frac{1}{\exp(\alpha)} \sum_{k = 0}^{\infty} \frac{1}{k!} \alpha^k f^k(x^*) = \frac{1}{\exp(\alpha)} \sum_{k = 0}^{\infty} \frac{1}{k!} \alpha^k x^* = \frac{1}{\exp(\alpha)} \left( \sum_{k = 0}^{\infty} \frac{1}{k!} \alpha^k \right) x^* = x^*,
\end{align*}
where in the last equality we applied the definition of $\exp(\alpha)$. This completes the proof.
\end{proof}

\begin{lemma}
Let $X$ be Euclidean space and $f( \cdot )$ a linear self-map on $X$. Then,
\begin{align*}
\exp(\alpha f(x)) = \frac{1}{\exp(\alpha)} \exp(\alpha A) x.
\end{align*}
\end{lemma}

\begin{proof}
Simply following the definitions, we obtain that
\begin{align*}
\exp(\alpha f(x)) = \frac{1}{\exp(\alpha)} \sum_{k = 0}^{\infty} \frac{1}{k!} \alpha^k f^k(x) = \frac{1}{\exp(\alpha)} \sum_{k = 0}^{\infty} \frac{1}{k!} \alpha^k A^k x = \frac{1}{\exp(\alpha)} \exp(\alpha A) x
\end{align*}
and this completes the proof.
\end{proof}

\noindent
Unsure how to call this generalization of exponentiation, we decided to call it {\em algorithmic boosting}.

\subsubsection{Variants of algorithmic boosting}

One may ponder if the exponential power series is essential in the definition of algorithmic boosting, and a moment of thought immediately suggests that it is not. We may generalize exponentially powered boosting as follows: Let $X$ be a convex set of vectors. Given a self-map $f(\cdot)$ on $X$, we define the powered correspondence $\mathbb{P}(f(\cdot))$ as a self-correspondence on $X$ such that $\mathbb{P}(f(x))$ is the $\omega$-limit set of the sequence
\begin{align*}
\left\{ \frac{1}{A_K} \sum_{k=0}^K \alpha_k f^k(x) \right\}_{K= 0}^{\infty},
\end{align*}
where $\{\alpha_k\}_{0}^{\infty}$ is a sequence of nonnegative scalars and $A_K = \sum_{k=0}^K \alpha_k$. Observe that since $X$ is convex, $\mathbb{P}(f(\cdot))$ is well-defined at every $x \in X$. This definition can be specialized based on the pattern of the sequence $\{\alpha_k\}_{0}^{\infty}$ of weights used in the averaging process: For example, if $\alpha_k=1, k = 0, 1, \ldots$, we obtain what we may call {\em geometrically powered boosting} (cf. geometric series), whereas if $\alpha_0 = 1$ and $\alpha_k = 1/k, k = 1, 2, \ldots$, we obtain what we may call {\em harmonically powered boosting} (cf. harmonic series). If there exists a natural number $m$ such that $\alpha_k = 0, k > m$, we obtain truncated series. We will see that these definitions also admit {\em variational interpretations}.\\

\noindent
Starting here, with the exception of Section \ref{exponentiatedpowermethod}, we use upper-case letters to denote vectors.

\subsection{Algorithmic boosting in game theory}

Algorithmic boosting theory in games is a generalization of a fundamental result in computational learning theory, namely, that the empirical average of iterated Hedge using a fixed learning rate converges to an approximate Nash equilibrium in a zero-sum game \citep{FreundSchapire2}. Let us recall that the fundamental solution concept in game theory is the Nash equilibrium, which is a fixed point of the best response correspondence. Nash's proof of the existence of an equilibrium in an $N$-person game is non-constructive. In this paper, we first focus on zero-sum games that readily admit a constructive Nash-equilibrium existence proof. We then provide a constructive existence proof of a symmetric Nash equilibrium in a symmetric bimatrix game. The problem of computing such an equilibrium is PPAD-complete \citep{CDT, Daskalakis} and, therefore, our constructive proof amounts to an existence result of the Nash equilibrium in a general game.

\subsubsection{Preliminaries in game theory}

Given a symmetric bimatrix game $(C, C^T)$, we denote the corresponding standard (probability) simplex by $\mathbb{X}(C)$. $\mathbb{\mathring{X}}(C)$ denotes the relative interior of $\mathbb{X}(C)$. The elements (probability vectors) of $\mathbb{X}(C)$ are called {\em strategies}. We call the standard basis vectors in $\mathbb{R}^n$ {\em pure strategies} and denote them by $E_i, i =1, \ldots, n$. {\em Symmetric Nash equilibria} are precisely those combinations of strategies $(X^*, X^*)$ such that $X^*$ satisfies $\forall Y \in \mathbb{X}(C) : X^* \cdot CX^* - Y \cdot CX^* \geq 0$. We call $X^*$ a {\em symmetric Nash equilibrium strategy}. A symmetric Nash equilibrium is guaranteed to always exist \citep{Nash2}. A simple fact is that $X^* \in \mathbb{X}(C)$ is a symmetric Nash equilibrium strategy if and only if $(CX^*)_{\max} - X^* \cdot CX^* = 0$. If $(CX^*)_{\max} - X^* \cdot CX^* \leq \epsilon$, $X^*$ is called an $\epsilon$-approximate equilibrium strategy. A symmetric bimatrix game $(C, C^T)$ is zero-sum if $C$ is antisymmetric, that is, $C = -C^T$.

\subsubsection{Computing a Nash equilibrium by averaging the powers of a map}

A plausible approach to compute a Nash equilibrium in this setting is to use Hedge. Given $C$, Hedge is given by map $T : \mathbb{X}(C) \rightarrow \mathbb{X}(C)$ where
\begin{align}
T_i(X) = X(i) \frac{\exp\{\alpha (CX)_i \}}{\sum_{j=1}^n X(j) \exp\{\alpha (CX)_j\}}, i =1, \ldots, n,\label{algebraicexpression}
\end{align}
and $\alpha$ is a parameter called the {\em learning rate}. To compute a symmetric Nash equilibrium of $(C, C^T)$ we can, for example, iterate Hedge using a fixed learning rate starting from an interior strategy $X^0 \in \mathbb{\mathring{X}}(C)$. However, a simple fact is that the sequence $\{ T^k(X^0) \}_{k=0}^{\infty}$ may not converge.

Algorithmic boosting theory factors at this critical moment to obtain convergence. Although the sequence $\{ T^k(X^0) \}_{k=0}^{\infty}$ may not converge, translating the aforementioned result of \cite{FreundSchapire2} in the language we are trying to develop in this paper, {\em geometrically powered boosting of iterated Hedge using a fixed learning rate $\alpha > 0$ converges to an $\epsilon$-approximate Nash equilibrium, where $\epsilon$ can be made as small as desired by choosing an accordingly small value of the learning rate $\alpha$}. What's more, \cite{AvramopoulosSSRN} shows in a recent contribution to this fundamental question that, under a diminishing learning rate schedule, what we may now call harmonically powered boosting of iterated Hedge converges to an exact Nash equilibrium of a symmetric zero-sum game. In this paper, we develop a set of algorithmic techniques that draw on {\em convex optimization} to devise an equilibrium fully polynomial time approximation scheme in a symmetric zero-sum game by averaging iterated Hedge. In this vein, we show that starting from the uniform strategy, to compute an $\epsilon$-approximate Nash equilibrium, our scheme requires at most
\begin{align*}
\floor*{\frac{\ln(n)}{\frac{\epsilon}{2}\ln\left(1+ \frac{\epsilon}{2} \right)}}
\end{align*} 
iterations. Our bound nearly exactly matches that of \cite{FreundSchapire2} but, since
\begin{align*}
\ln\left(1+ \frac{\epsilon}{2} \right) > \frac{\frac{\epsilon}{2}}{1 + \frac{\epsilon}{2}}, \epsilon > 0,
\end{align*}
it is better. We also leverage our techniques to obtain matching bounds under a significantly more numerically stable version of Hedge that is obtained by generalizing \eqref{algebraicexpression} as $T: \mathbb{X}(C) \times \mathbb{X}(C) \rightarrow \mathbb{X}(C)$, where
\begin{align}
T_i(X | Z) = X(i) \frac{\exp\{\alpha (CZ)_i \}}{\sum_{j=1}^n X(j) \exp\{\alpha (CZ)_j\}}, i =1, \ldots, n.\label{generalization}
\end{align}
We discuss our approach below. Before that let's get to our result on the continuous limit of Hedge.

\subsubsection{Computing a Nash equilibrium by averaging an orbit of a flow}

Taking the long-run average of an orbit of a differential equation is the {\em variational analogue of algorithmic boosting}. In this paper, we also consider a variational analog of Hedge, namely, the {\em replicator dynamic} \citep{TaylorJonker}. We first show that Hedge is a convergent and consistent numerical integrator of the replicator dynamic. We then prove that, in a symmetric zero-sum game, the $\omega$-limit set of the average of every orbit of this dynamic starting in the interior of the simplex is a symmetric Nash equilibrium strategy. Our result informs a line of research regarding the divergent behavior of the replicator dynamic in zero-sum games \citep{Mertikopoulos, Biggar2} in an elegant fashion: Although orbits may diverge (for example, they can cycle \citep{Akin}), the long-run average converges. The theory we lay out in this paper further informs a line of research regarding the divergent behavior of evolutionary dynamics in general \citep{Flokas, Milionis}: Although the orbits themselves may fail to converge, we stipulate that, in the fashion of our result on the replicator dynamic but also in the fashion of \citep{FreundSchapire2}, a variety of carefully constructed averages may converge.

\subsubsection{Computing a Nash equilibrium by clairvoyant averaging}

As a case in point, a contribution of this paper is an algorithm for computing a symmetric Nash equilibrium in a symmetric bimatrix game that gives a constructive proof of existence of such an equilibrium. To obtain that algorithm, we construct the Euler approximation of \eqref{generalization} to obtain
\begin{align}
T_i(X | Z) = X(i) (1 + \alpha ((CZ)_i - X \cdot CZ)), i =1, \ldots, n\label{Euler_approximation}
\end{align}
and, in a sense, generalizing the clairvoyant regret minimization paradigm \citep{Clairvoyant}, we require that
\begin{align}
X^{K+1}(i) \equiv \frac{1}{A_K} \sum_{k=0}^K \alpha_k Z^k(i) = X^K(i) (1 + \alpha_K ((CZ^K)_i - X \cdot CZ^K)), i =1, \ldots, n\label{Euler_approximation}
\end{align}
for $K = 0, 1, \ldots,$. We show that there is a sequence of learning rates $\{ \alpha_K \}$ and ``multipliers'' $\{ Z^K \}$ such that the previous system of equations always has a solution and that the sequence $\{ X^K \}$ converges to an approximate symmetric Nash equilibrium where the approximation error can be made arbitrarily small by choosing a correspondingly small upper bound on the learning rates.

\subsection{Implementing algorithmic boosting}

%The material benefits from applying algorithmic boosting theory cannot usually be immediately realized and the typical case is that extra algorithmic effort is required to reap the fruit. To a great extent, implementing algorithmic boosting methods as algorithms is the most technically exciting part of our paper. Let us, for example, consider the aforementioned exponentiated power method.

Once a sequence of iterates is shown to converge to a desired fixed point, there is a general-purpose technique to accelerate that convergence, namely, the {\em vector epsilon algorithm} \citep{Pade, EpsilonAlgorithm}. In principle, the convergent sequence can be obtained either by the iterative application of a discrete map or the discretization of a differential equation. Naturally the technique applies to convergent sequences of averages. The vector epsilon algorithm can be formalized, for example, using Pad\'e approximation theory. To keep our paper focused, we first discuss Pad\'e approximants as they apply to exponentiating scalars and to the computation of the matrix exponential followed by results on the implementation of the exponentiation in Hedge. Such implementations of algorithmic boosting are one of the most exciting parts of this research.

\subsubsection{Using Pad\'e approximation theory to compute the matrix exponential}

One approach to approximately compute the matrix exponential is to truncate its Taylor series. In fact, this is just one out a big list of methods \citep{Moler}. A method that is often used in practice is to compute the Pad\'e approximant. In general, Pad\'e approximations approximate a function by a rational function of a given order. A $[p/q]$ Pad\'e approximant is a ratio of a polynomial of degree $p$ over a polynomial of degree $q$. For example, the $[3/3]$ Pad\'e approximant of $\exp(x)$, denoted by $\exp_{3/3}(x)$, is
\begin{align*}
\exp_{3/3}(x) = \frac{120 + 60x + 12 x^2 + x^3}{120 - 60x +12 x^2 - x^3}.
\end{align*}
Pad\'e approximations of the matrix exponential follow the same principle \citep{Arioli, Higham}. Since Pad\'e approximation theory generalizes to non-convergent series, we believe there is fertile ground to apply this theory to implement general algorithmic boosting methods.

\subsubsection{Using relative entropy programming to implement Hedge}

Once a general framework for the implementation of the exponential function has been set into place, it becomes immediately apparent that the exponentiation operation used in the Hedge map may entail complexities that require careful attention. \cite{Krichene} show that the Hedge map is a dual formulation of a convex optimization problem (with the same solution). The exact solution of this problem can be computed to any desirable precision in polynomial time but this implementation requires solving a {\em relative entropy program} \citep{Chandra} using an {\em interior-point convex-programming method} \citep{Nesterov}. In a numerical experiment we show that there is a large discrepancy between the simple algebraic implementation of Hedge as that is suggested in equation \eqref{algebraicexpression} and the robust method that uses convex programming.

\subsubsection{Approximately solving a zero-sum game using a numerically stable Hedge}

It is natural then to ask, for example, to what extent the bounds of \citep{FreundSchapire2} render practical algorithms to approximately solve zero-sum games (and, therefore, also linear programs) using inexpensive implementations of the exponential function. Our intuition from the aforementioned numerical experiment suggests that the vanilla implementation of the exponential function may not always suffice to obtain the desired performance. Our numerical experiments in Matlab show that there is a deviation in the evolution of the Hedge map under inexpensive exponentiation and under the robust relative entropy programming implementation {\em even in the rock-paper-scissors game using a small learning rate}. In this paper, we show how to restore numerical stability, using inexpensive exponentiation in the same numerical environment, first by generalizing Hedge as in equation \eqref{generalization} and then requiring the ``multipliers'' $Z$ to only correspond to pure strategies ($Z \in \{E_1, \ldots, E_n\}$). Our algorithm is straightforward to implement, it is conceptually simpler than Hedge as $CZ$ only correspond to column vectors of $C$, it is numerically stable since the exponentials can be precomputed to any desired accuracy ahead of time, and as we prove in the sequel, it has the same theoretical performance as Hedge (in that the bounds are identical).

\subsection{Other related work}

%This paper was motivated by a result by one of the authors \citep{AvramopoulosGEB} that the average orbit of an exponentiated version of the replicator dynamic converges to a Nash equilibrium in a symmetric bimatrix game. In this paper, we go beyond that paper in trying to create a sound theoretical foundation for understanding exponentiation and time-averages for the benefit of designing {\em fixed-point computation systems}. In particular, in this paper, we provide a sound algorithmic foundation for studying the problem of computing fixed points by averaging the powers of a map. 

Closely related to our algorithmic boosting theory is {\em ergodic theory} that systematically studies {\em conservative systems} from a similar perspective. The problems we study in this paper are not, in general, conservative. In such a sense, algorithmic boosting theory is a generalization of ergodic theory. We believe that studying ergodic systems from an algorithmic boosting perspective, for example, exponentiating the Hamiltonian operator, is an interesting direction for future work. 

The precise role of exponentiation in the multiplicative weights update method is not well-understood. An aspect is explored by \cite{Pelillo} who report an empirical finding that using exponentiation (in the same fashion as that is used in the multiplicative weights update method) in quadratic optimization significantly increases the convergence rate over more elementary algorithms that obviate exponentiation: For example, Hedge is faster than the {\em replicator dynamic} or the {\em discrete-time replicator dynamic}. Our independently performed experiments confirm this.

Our variational perspective on algorithmic boosting is in the paradigm of \cite{Jordan} who initiated the study of acceleration methods in optimization from a variational perspective. In fact, that Hedge is a discretization of the replicator dynamic was an idea drawn from that paradigm.

Although our paper is not the first to propose using the matrix exponential in link analysis \citep{Miller}, to the extent of our knowledge, our paper is the first that observes that using the matrix exponential accelerates the classical power method and analyzes the precise impact on the convergence rate, in particular, that the convergence rate increases by an exponential factor.

Related to link analysis is the problem of computing probabilities of random walks \citep{Aleliunas} where the technique of averaging powers of a linear map has recently factored prominently in efficient algorithms that approximate these probabilities \citep{Vadhan}.

Let us also note that different from previous constructive proofs of existence of the Nash equilibrium, in particular, the Lemke-Howson algorithm and the algorithm of \cite{LMM} that is related to \citep{Lipton94Simple}, our existence proof is by a {\em convergent sequence}, which, for example, renders it amenable to acceleration by the epsilon algorithm \citep{EpsilonAlgorithm}. 

Let us, finally, note there is an algorithm for approximately solving the PPAD-complete {\em fractional stable paths problem} \citep{Kintali} that does not require averaging \citep{Kintali-FSPP}.

\subsection{Overview of the rest of this paper}

In the next section, we enumerate our contributions in detail whereas in Section \ref{exponentiatedpowermethod} we present and analyze our exponentiated power method. In Sections \ref{HedgeBoosting} and \ref{ReplicatorBoosting} we analyze algorithmic boosting in game theory (both in discrete and continuous time). Our clairvoyant averaging principle is presented and analyzed in Section \ref{constructive}. Sections \ref{NumericalHedge} and \ref{DoubleExponentiation} discuss implementations of the Hedge map. Section \ref{Pagerank-Nash} draws a relationship between Pagerank and Nash equilibria. Finally, Section \ref{Conclusion} concludes.

\section{Our detailed contributions}
\label{detailed}

In this paper, we develop a theory of {\em algorithmic exponentiation} based on the {\em matrix exponential} and its generalization to {\em nonlinear} fixed-point computation systems. Our contributions are:

\begin{enumerate}

\item We show that exponentiating convergent {\em iterated linear maps}, such as the Google matrix in PageRank, using a positive learning rate gives an exponential increase in their convergence rate, which increases with the learning rate. The acceleration depends on the relative difference between the dominant and the radius of the second dominant eigenvalue. We further show this method to be practical: Computation of the exact exponential matrix can be obviated by approximating it with the truncated series, while obtaining comparable performance.

\item A slight generalization of algorithmic exponentiation is {\em algorithmic boosting}, the principle of computing fixed points by averaging the powers of an iterated map. Focusing on game theory, we slightly improve the Nash equilibrium approximation error of the Hedge map in a symmetric zero-sum game by better analysis. Our improvement is perhaps marginal but we further leverage this analysis to obtain a new Nash equilibrium computation algorithm in a symmetric zero-sum game that is related to the archetypical Hedge map.

\item We show that Hedge is a consistent and convergent numerical integrator of the (continuous-time) replicator dynamic.

\item We consider a {\em variational approach} to algorithmic boosting in the same game-theoretic setting and show that every limit point of the long-run average of the replicator dynamic (in a symmetric zero-sum game) is a Nash equilibrium. 

\item We introduce a ``clairvoyant averaging principle'' using the Euler approximation of a generalization of the Hedge map, which we leverage to obtain a constructive existence proof of a symmetric Nash equilibrium in a symmetric bimatrix game. Since this problem is PPAD-complete, we, thus, also obtain a constructive existence proof of Brouwer fixed points.

\item We observe that there exists a polynomial-time algorithm to compute the iterates of the Hedge map, using relative entropy programming, and demonstrate that there can be a substantial difference in numerical performance between this robust implementation and the naive one.

\item We devise a simplified Hedge that attains identical equilibrium approximation bounds in symmetric zero-sum games while at the same time being more numerically stable than the Hedge map itself. The analysis of this algorithm's bounds draws on the same techniques we introduce to improve the Hedge equilibrium approximation bound (as mentioned above).

\item We establish a relationship between dominant eigenvector computations and Nash equilibria.

\end{enumerate}

\section{Algorithmically boosting convergent linear fixed-point iterators}
\label{exponentiatedpowermethod}

In this section, we consider the power iteration method for computing a dominant eigenvector of an $n \times n$ matrix $A$, which starts with a random vector $x_0$ and recursively applies the linear map corresponding to $A$ to $x_0$. This process converges to the dominant eigenvector. In this section, we first analyze the following idea: Instead of using $A$ itself in the power iterations, use $\exp(\alpha A)$, where $\alpha > 0$ and $\exp( \cdot )$ is the matrix exponential. We show that our proposed method gives an exponential increase in the convergence rate (since exponentiation exponentiates the eigenvalues). Such an increase is obtained by an exact computation of the matrix exponential. We also analyze the precise increase in the convergence rate obtained by truncating the matrix exponential series to order $m \geq 1$ wherein the upper incomplete gamma function factors in the convergence ratio of the corresponding geometric convergence rate. Let us start by defining the power iteration method.

\subsection{The power iteration method}

Given an $n \times n$ matrix $A$ assume that its eigenvalues are ordered such that $|\lambda_1| > |\lambda_2| \geq \cdots \geq |\lambda_n|$. The power iteration method (for example, see \citep[Chapter 7.3]{Golub}) computes the {\em dominant} (i.e., largest in modulus) eigenvalue $\lambda_1$ and corresponding (dominant) eigenvector. To that end, it starts with a $n \times 1$ vector $x_0$ and iteratively generates a sequence $\{ x_k\}$ using the recurrence relation
\begin{align*}
x_{k+1} = \frac{A x_k}{\| A x_k \|}, k=0,1,2,\ldots
\end{align*}
The sequence $\{x_k\}$ converges to the dominant eigenvector provided that $x_0$ is not orthogonal to the left eigenvector corresponding to $\lambda_1$. Furthermore, under these assumptions, the sequence
\begin{align*}
\left\{ \frac{x_k^T A x_k}{x_k^T x_k} \right\}_{k=0}^{\infty}
\end{align*}
converges to $\lambda_1$. A typical case in the practical application of the method is that $\lambda_1$ is equal to one.

\subsection{Analysis when $A$ is diagonalizable}

Let us first consider (as a warmup) the power iteration methods assuming $A$ is {\em diagonalizable}. We consider the simple power iteration method first followed by the exponentiated power method.

\subsubsection{The simple power iteration method}

So let us assume $A$ is diagonalizable, let $\lambda_1,  \lambda_2, \ldots, \lambda_n$ be the $n$ eigenvalues of $A$ (counted with multiplicity) and let $u_1, u_2, \ldots u_n$ be the corresponding eigenvectors. Suppose that $\lambda_1$ is the dominant eigenvalue, so that $|\lambda_1| > |\lambda_j|$ for $j > 1$. We then have
\begin{align*}
x_0 = c_1 u_1 + c_2 u_2 + \cdots + c_n u_n
\end{align*}
and
\begin{align*}
A^k x_0 &= c_1 A^k u_1 + c_2 A^k u_2 + \cdots + c_n A^k u_n\\
&= c_1 \lambda_1^k u_1 + c_2 \lambda_2^k u_2 + \cdots + c_n \lambda_n^k u_n\\
&= c_1 \lambda_1^k \left( u_1 + \frac{c_2}{c_1} \left( \frac{\lambda_2}{\lambda_1} \right)^k u_2 + \cdots + \frac{c_n}{c_1} \left( \frac{\lambda_n}{\lambda_1} \right)^k u_n \right).
\end{align*}
Observe now that since
\begin{align*}
\left| \frac{\lambda_j}{\lambda_1} \right| < 1, j > 1
\end{align*}
we obtain that
\begin{align*}
A^k x_0 \rightarrow c_1 \lambda_1^k u_1
\end{align*}
and since 
\begin{align*}
x_{k+1} = \frac{A x_k}{\| A x_k \|}, k=0,1,2,\ldots
\end{align*}
we obtain that the sequence $\{ x_k \}$ converges to (a multiple of) the eigenvector $u_1$. The convergence is {\em geometric} with ratio
\begin{align*}
\left| \frac{\lambda_2}{\lambda_1} \right|,
\end{align*}
where $\lambda_2$ is the second dominant eigenvalue.

\subsubsection{The fully exponentially powered iteration method}

Like the power method, our {\em exponentiated power method} computes the dominant eigenvalue and corresponding eigenvector of an $n \times n$ matrix $A$ under the same conditions as the aforementioned power method, using instead the {\em matrix exponential} of $\alpha A$, which we denote by $\exp(\alpha A)$, where
\begin{align*}
\exp(A) = I + A + \frac{1}{2} A^2 + \cdots,
\end{align*}
and $\alpha > 0$ is a parameter we call the {\em learning rate}. The exponentiated power iteration starts with a $n \times 1$ vector $x_0$ and iteratively generates a sequence $\{ x_k\}$ using the recurrence relation
\begin{align*}
x_{k+1} = \frac{\exp(\alpha A) x_k}{\| \exp(\alpha A) x_k \|}, k=0,1,2,\ldots
\end{align*}
Let us prove that the sequence $\{x_k\}$ converges to the dominant eigenvector under the same conditions that the power method converges. To that end, let $\lambda_1,  \lambda_2, \ldots, \lambda_n$ be the $n$ eigenvalues of $A$ (counted with multiplicity) and let $u_1, u_2, \ldots u_n$ be the corresponding eigenvectors. Suppose that $\lambda_1 > 0$ is the dominant eigenvalue, so that $\lambda_1 > |\lambda_j|$ for $j > 1$. We then have
\begin{align*}
x_0 = c_1 u_1 + c_2 u_2 + \cdots + c_n u_n
\end{align*}
and
\begin{align*}
\exp(\alpha A)^k x_0 &= c_1 \exp(\alpha A)^k u_1 + c_2 \exp(\alpha A)^k u_2 + \cdots + c_n \exp(\alpha A)^k u_n\\
&= c_1 \exp(\alpha \lambda_1)^k u_1 + c_2 \exp(\alpha \lambda_2)^k u_2 + \cdots + c_n \exp(\alpha \lambda_n)^k u_n\\
&= c_1 \exp(\alpha \lambda_1)^k \left( u_1 + \frac{c_2}{c_1} \left( \exp(\alpha (\lambda_2 - \lambda_1) \right)^k u_2 + \cdots + \frac{c_n}{c_1} \left( \exp(\alpha (\lambda_n - \lambda_1) \right)^k u_n \right),
\end{align*}
where we have used the simple property that if $u$ is an eigenvector of $A$ corresponding to eigenvalue $\lambda$, then it is also an eigenvector of $\exp(A)$ with corresponding eigenvalue $\exp(\lambda)$. Observe now that, assuming $\lambda_1 > 0$, since
\begin{align*}
\frac{\exp(\alpha \lambda_j)}{\exp(\alpha \lambda_1)} < 1, j > 1
\end{align*}
we obtain that
\begin{align*}
\exp(\alpha A)^k x_0 \rightarrow c_1 \exp(\alpha \lambda_1)^k u_1
\end{align*}
and since 
\begin{align*}
x_{k+1} = \frac{\exp(\alpha A) x_k}{\| \exp(\alpha A) x_k \|}, k=0,1,2,\ldots
\end{align*}
we obtain that the sequence $x_k$ converges to (a multiple of) the eigenvector $v_1$. The convergence is {\em geometric} with ratio at least
\begin{align*}
\exp(\alpha (|\lambda_2| - \lambda_1)),
\end{align*}
where $\lambda_2$ is the second dominant eigenvalue.

\subsection{Analysis in the general case}

In the general case, the fully exponentiated power method gives an exponential speedup in the rate by which the sinusoid of the angle between the iterates and the dominant eigenvector goes to zero. \citep[Chapter 7]{Arbenz} shows that in the simple power method the sinusoid of this angle (between the iterates and the dominant eigenvector) converges to zero geometrically with a rate equal to
\begin{align*}
\left| \frac{\lambda_2}{\lambda_1} \right|.
\end{align*}
In particular, it is shown that:

\begin{proposition}
\label{arbenzproposition}
Let the eigenvalues of the $n \times n$ matrix $A$ be arranged such that $|\lambda_1| > |\lambda_2| \geq \cdots \geq |\lambda_n|$. Furthermore, Let $u_1$ and $v_1$ be the right and left eigenvectors of $A$ corresponding to $\lambda_1$, respectively. Then the sequence of vectors generated by the power iteration method converges to $u_1$ in the sense that
\begin{align*}
\sin(\measuredangle(x_k, u_1)) \leq c \left| \frac{\lambda_2}{\lambda_1} \right|^k
\end{align*}
provided $v_1^* x_0 \neq 0$.
\end{proposition}

Our exponentiated power iteration exponentially increases the convergence rate and this increase depends on the learning rate: As the learning rate $\alpha > 0$ increases the convergence rate increases likewise. Our main result in this direction is the following theorem, omitting the proof, which is very similar to the proof of Proposition \ref{arbenzproposition} shown in \citep[Chapter 7]{Arbenz}.

\begin{theorem}
Let the eigenvalues of the $n \times n$ matrix $A$ be arranged such that $\lambda_1 > |\lambda_2| \geq \cdots \geq |\lambda_n|$. Furthermore, Let $u_1$ and $v_1$ be the right and left eigenvectors of $A$ corresponding to $\lambda_1$, respectively. Then the sequence of vectors generated by the exponentiated power iteration method with learning rate $\alpha > 0$ converges to $u_1$ in the sense that
\begin{align*}
\sin(\measuredangle(x_k, u_1)) \leq c (\exp(\alpha (|\lambda_2| - \lambda_1))^k
\end{align*}
provided $v_1^* x_0 \neq 0$.
\end{theorem}

\subsection{Truncating the exponential series}

The matrix exponential can rarely be computed exactly and it is typically approximated using a Pad\'e approximant. Computing a Pad\'e approximant can be expensive requiring, for example, a matrix inversion operation. If the goal is to iterate only a few times to compute an approximate dominant eigenvector, the decision whether to invoke the Pad\'e approximation over the simple power iteration method can be analytically complex. A conceptually and analytically simpler approach is to {\em truncate the exponential series}. Using the first $m$ terms of the exponential series to iterate on, that is, using the power iteration method with operand matrix
\begin{align*}
\bar{A} = \sum_{k = 0}^m \frac{1}{k!} (\alpha A)^k
\end{align*}
it is easily seen using the previous convergence analysis that the convergence rate is geometric with ratio
\begin{align*}
\left| \frac{\sum_{k=0}^m \frac{1}{k!} (\alpha \lambda_2)^k}{\sum_{k=0}^m \frac{1}{k!} (\alpha \lambda_1)^k} \right| = \left| \frac{\Gamma (m+1, \alpha \lambda_2)}{\Gamma (m+1, \alpha \lambda_1)} \right| \exp(\alpha(|\lambda_2| - \lambda_1)
\end{align*}
where $| \cdot |$ is the absolute value, $\Gamma(\cdot, \cdot)$ is the upper incomplete gamma function, and $\lambda_1 > 0, \lambda_2$ are the dominant and second dominant eigenvalues. Therefore, even just a second order approximation ($m=2$) can give a significant acceleration in the power iteration method especially as $\alpha$ increases.

\section{Algorithmically boosting non-convergent fixed-point iterators}
\label{HedgeBoosting}

We have previously claimed that algorithmically powered boosting can convert a non-convergent map to one that converges. In this section, we add rigor to the previous discussion on this matter. Toward making this important point, we consider a symmetric zero-sum game, that is, a symmetric bimatrix game $(C, C^T)$ such that the payoff matrix $C$ is antisymmetric (in that $C^T = - C$). We aim to compute or approximate a Nash equilibrium in such a game by repeatedly applying Hedge starting from an interior strategy $X \in \mathbb{\mathring{X}}(C)$. After preliminary results, we prove that iterated Hedge fails to converge in this setting. We then prove that averaging iterated Hedge indeed converges. In particular, we give a fully polynomial time approximation scheme for computing an $\epsilon$-approximate symmetric Nash equilibriums strategy. Our proof techniques in this vein are interesting and novel. In Section \ref{DoubleExponentiation}, we leverage our proof techniques to analyze a doubly exponentiated Hedge map.

%We close this section by discussing how results obtained in this fashion can render exact polynomial-time algorithms for computing Nash equilibria using the theory of Pad\'e approximations.

\subsection{Preliminary properties of Hedge}

Let us repeat the Hedge map for convenience:
\begin{align}
T_i(X) = X(i) \frac{\exp \left\{ \alpha (CX)_i \right\}}{\sum_{j=1}^n X(j) \exp\{\alpha (CX)_j \}}, \quad i = 1, \ldots, n.\label{main_exp}
\end{align}
In this section, $\mathcal{C}(X)$ denotes the carrier of $X$, that is, the pure strategies that support $X$. Furthermore, $\mathbb{\hat{C}}$ denotes the class of payoff matrices $C$ whose entries lie in the range $[0, 1]$.

\subsubsection{Relative entropy (or Kullback-Leibler divergence)}

Our analysis of Hedges relies on the relative entropy function between probability distributions (also called {\em Kullback-Leibler divergence}). The relative entropy between the $n \times 1$ probability vectors $P > 0$ (that is, for all $i = 1, \ldots, n$, $P(i) > 0$) and $Q > 0$ is given by 
\begin{align*}
RE(P, Q) \doteq \sum_{i=1}^n P(i) \ln \frac{P(i)}{Q(i)}.
\end{align*}
However, this definition can be relaxed: The relative entropy between $n \times 1$ probability vectors $P$ and $Q$ such that, given $P$, for all $Q \in \{ \mathcal{Q} \in \mathbb{X} | \mathcal{C}(P) \subset \mathcal{C}(\mathcal{Q}) \}$, where $\mathbb{X}$ is a probability simplex of appropriate dimension, is
\begin{align*}
RE(P, Q) \doteq \sum_{i \in \mathcal{C}(P)} P(i) \ln \frac{P(i)}{Q(i)}.
\end{align*}
We note the well-known properties of the relative entropy \cite[p.96]{Weibull} that {\em (i)} $RE(P, Q) \geq 0$, {\em (ii)} $RE(P, Q) \geq \| P - Q \|^2$, where $\| \cdot \|$ is the Euclidean distance, {\em (iii)} $RE(P, P) = 0$, and {\em (iv)} $RE(P, Q) = 0$ iff $P = Q$. Note {\em (i)} follows from {\em (ii)} and {\em (iv)} follows from {\em (ii)} and {\em (iii)}.

\subsubsection{The convexity lemma}

The following lemma generalizes \cite[Lemma 2]{FreundSchapire2}.

\begin{lemma}
\label{convexity_lemma}
Let $T$ be as in \eqref{main_exp}. Then
\begin{align*}
\forall X \in \mathbb{\mathring{X}}(C) \mbox{ } \forall Y \in \mathbb{X}(C) : RE(Y, T(X)) \mbox{ is a convex function of }\alpha.
\end{align*}
Furthermore, unless $X$ is a fixed point, $RE(Y, T(X))$ is a strictly convex function of $\alpha$.
\end{lemma}

\begin{proof}
We have
{\allowdisplaybreaks
\begin{align*}
\frac{d}{d \alpha}& RE(Y, \hat{X}) = \\
  &= \frac{d}{d \alpha} \left( \sum_{i \in \mathcal{C}(Y)} Y(i) \ln\left( \frac{Y(i)}{\hat{X}(i)} \right) \right)\\
  &= \frac{d}{d \alpha} \left( \sum_{i \in \mathcal{C}(Y)} Y(i) \ln\left( Y(i) 
  \cdot \frac{\sum_{j = 1}^n X(j) \exp\{ \alpha (CX)_j \}}{ X(i) \exp\{ \alpha (CX)_i \}} \right) \right)\\
  &= \frac{d}{d \alpha} \left( \sum_{i \in \mathcal{C}(Y)} Y(i) \ln\left( \frac{\sum_{j = 1}^n X(j) \exp\{ \alpha (CX)_j \}}{X(i)\exp\{ \alpha (CX)_i \}} \right) \right)\\
  &= \sum_{i \in \mathcal{C}(Y)} Y(i) \frac{d}{d \alpha} \left( \ln\left( \frac{\sum_{j = 1}^n X(j) \exp\{ \alpha (CX)_j \}}{X(i)\exp\{ \alpha (CX)_i \}} \right) \right).
\end{align*}}
Furthermore, using $(\cdot)'$ as alternative notation (abbreviation) for $d/d\alpha(\cdot)$,
\begin{align*}
\frac{d}{d \alpha} \left( \ln\left( \frac{\sum_{j = 1}^n X(j) \exp\{ \alpha (CX)_j \}}{X(i)\exp\{ \alpha (CX)_i \}} \right) \right) = \frac{X(i)\exp\{ \alpha (CX)_i \}}{\sum_{j = 1}^n X(j) \exp\{ \alpha (CX)_j \}} \left( \frac{\sum_{j = 1}^n X(j) \exp\{ \alpha (CX)_j \}}{X(i)\exp\{ \alpha (CX)_i \}} \right)'
\end{align*}
and
\begin{align*}
\left( \frac{\sum_{j = 1}^n X(j) \exp\{ \alpha (CX)_j \}}{X(i)\exp\{ \alpha (CX)_i \}} \right)'  &= \frac{\sum_{j = 1}^n X(j) (CX)_j \exp\{ \alpha (CX)_j \} X(i) \exp\{ \alpha (CX)_i \}}{\left( X(i) \exp\{ \alpha (CX)_i \} \right)^2} -\\
  &- \frac{X(i) (CX)_i \sum_{j = 1}^n X(j) \exp\{ \alpha (CX)_j \} \exp\{ \alpha (CX)_i \}}{\left( X(i) \exp\{ \alpha (CX)_i \} \right)^2} =\\
  &= \frac{\sum_{j = 1}^n X(j) (CX)_j \exp\{ \alpha (CX)_j \} \}}{ X(i) \exp\{ \alpha (CX)_i \} } -\\
  &- \frac{(CX)_i \sum_{j = 1}^n X(j) \exp\{ \alpha (CX)_j \} \}}{ X(i) \exp\{ \alpha (CX)_i \} }.
\end{align*}
Therefore,
\begin{align}
\frac{d}{d \alpha} RE(Y, \hat{X}) = \frac{\sum_{j = 1}^n X(j) (CX)_j \exp\{ \alpha (CX)_j \}}{\sum_{j = 1}^n X(j) \exp\{ \alpha (CX)_j \}} - Y \cdot CX.\label{valentine}
\end{align}
Furthermore,
\begin{align*}
\frac{d^2}{d \alpha^2} RE(Y, \hat{X}) &= \frac{\left( \sum_{j = 1}^n X(j) ((CX)_j)^2 \exp\{ \alpha (CX)_j \} \right) \left( \sum_{j = 1}^n X(j) \exp\{ \alpha (CX)_j \} \right)}{\left( \sum_{j = 1}^n X(j) \exp\{ \alpha (CX)_j \} \right)^2} -\\
  &- \frac{\left( \sum_{j = 1}^n X(j) ((CX)_j) \exp\{ \alpha (CX)_j \} \right)^2}{\left( \sum_{j = 1}^n X(j) \exp\{ \alpha (CX)_j \} \right)^2}.
\end{align*}
Jensen's inequality implies that
\begin{align*}
\frac{\sum_{j = 1}^n X(j) ((CX)_j)^2 \exp\{ \alpha (CX)_j \}}{\sum_{j = 1}^n X(j) \exp\{ \alpha (CX)_j \}} \geq \left( \frac{\sum_{j = 1}^n X(j) ((CX)_j) \exp\{ \alpha (CX)_j \}}{\sum_{j = 1}^n X(j) \exp\{ \alpha (CX)_j \}} \right)^2,
\end{align*}
which is equivalent to the numerator of the second derivative being nonnegative as $X$ is a probability vector. Note that the inequality is strict unless 
\begin{align*}
\forall i, j \in \mathcal{C}(X) : (CX)_i = (CX)_j.
\end{align*}
This completes the proof.
\end{proof}

\subsubsection{A version of the convexity lemma}

The following lemma is an analogue of \cite[Lemma 8.2.1, p. 471]{ConvexAnalysis}.

\begin{lemma}
\label{convexity_lemma_2}
Let $C \in \mathbb{\hat{C}}$. Then, for all $Y \in \mathbb{X}(C)$ and for all $X \in \mathbb{\mathring{X}}(C)$, we have that
\begin{align*}
\forall \alpha > 0 : RE(Y, T(X)) \leq RE(Y, X) - \alpha (Y-X) \cdot CX + \alpha (\exp\{\alpha\} - 1) \bar{C},
\end{align*}
where $\bar{C} > 0$ is a scalar that can be chosen independent of $X$ and $Y$.
\end{lemma}

\begin{proof}
Since, by Lemma \ref{convexity_lemma}, $RE(Y, T(X)) - RE(Y, X)$ is a convex function of $\alpha$, we have by the aforementioned secant inequality that, for $\alpha > 0$,
\begin{align}
RE(Y, T(X)) - RE(Y, X) \leq \alpha \left( RE(Y, T(X)) - RE(Y, X) \right)' = \alpha \cdot \frac{d}{d \alpha} RE(Y, T(X)).\label{ooone}
\end{align}
Straight calculus (cf. Lemma \ref{convexity_lemma}) implies that
\begin{align*}
\frac{d}{d \alpha} RE(Y, T(X)) = \frac{\sum_{j = 1}^n X(j) (CX)_j \exp\{ \alpha (CX)_j \}}{\sum_{j = 1}^n X(j) \exp\{ \alpha (CX)_j \}} - Y \cdot CX.
\end{align*}
Using Jensen's inequality in the previous expression, we obtain 
\begin{align}
\frac{d}{d \alpha} RE(Y, T(X)) \leq \frac{\sum_{j = 1}^n X(j) (CX)_j \exp\{ \alpha (CX)_j \}}{\exp\{ \alpha X \cdot CX \}} - Y \cdot CX.\label{vbvbvb}
\end{align}
Note now that
\begin{align}
\exp\{ \alpha x \} \leq 1 + (\exp\{ \alpha \} - 1) x, x \in [0, 1],\label{freund_schapire_in}
\end{align}
an inequality used in \cite[Lemma 2]{FreundSchapire2}. Using $C \in \mathbb{\hat{C}}$, \eqref{vbvbvb} and \eqref{freund_schapire_in} imply that
\begin{align*}
\frac{d}{d \alpha} RE(Y, T(X)) \leq \frac{X \cdot CX}{\exp\{ \alpha X \cdot CX \}} - Y \cdot CX + \left( \exp\{ \alpha \} - 1 \right) \frac{\sum_{j=1}^n X(j) (CX)_j^2}{\exp\{ \alpha X \cdot CX \}}
\end{align*}
and since $\exp\{\alpha X \cdot CX\} \geq 1$ (again by the assumption that $C \in \mathbb{\hat{C}}$), we have
\begin{align*}
\frac{d}{d \alpha} RE(Y, T(X)) \leq X \cdot CX - Y \cdot CX + \left( \exp\{ \alpha \} - 1 \right) \sum_{j=1}^n X(j) (CX)_j^2.
\end{align*}
Choosing $\bar{C} = \max \left\{\sum X(j) (CX)_j^2 \right\}$ and combining with \eqref{ooone} yields the lemma.
\end{proof}

\subsubsection{An instability lemma}

The following lemma is crucial in deriving divergence results on multiplicative weights in general. We can prove it in two ways, one invoking the aforementioned convexity lemma and the other by simply invoking Jensen's inequality. We show both proofs.

\begin{lemma}
\label{cool_hedge_2}
Let $Y, X \in \mathbb{X}(C)$ such that $\mathcal{C}(Y) \subseteq \mathcal{C}(X)$ and such that $Y \neq X$. If $X$ is not a fixed point, then
\begin{align*}
X \cdot CX - Y \cdot CX \geq 0 \Rightarrow \forall \alpha > 0 : RE(Y, T(X)) - RE(Y, X) > 0.
\end{align*}
\end{lemma}

\begin{proof}[First proof of Lemma \ref{cool_hedge_2}]
Let $\hat{X} = T(X)$.We have, by Jensen's inequality, that
{\allowdisplaybreaks
\begin{align}
RE(Y, \hat{X}) - RE(Y, X) &= \sum_{i \in \mathcal{C}(Y)} Y(i) \ln \frac{Y(i)}{\hat{X}(i)} - \sum_{i \in \mathcal{C}(Y)}  Y(i) \ln \frac{Y(i)}{X(i)}\notag\\
                       &= - \sum_{i \in \mathcal{C}(Y)} Y(i) \ln \hat{X}(i) + \sum_{i \in \mathcal{C}(Y)} Y(i) \ln X(i)\notag\\
                       &= \sum_{i \in \mathcal{C}(Y)} Y(i)  \ln \frac{X(i)}{\hat{X}(i)}\notag\\
                       &= \sum_{i \in \mathcal{C}(Y)} Y(i)  \ln \left( \frac{\sum_{j=1}^n X(j) \exp\{ \alpha E_j \cdot CX \}}{\exp\{\alpha E_i \cdot CX\}} \right)\notag\\
                       &= \ln \left( \sum_{j=1}^n X(j) \exp\{ \alpha E_j \cdot CX \} \right) - \sum_{i \in \mathcal{C}(Y)} Y(i)  \ln\left(\exp\{ \alpha E_i \cdot CX \} \right)\notag\\
                       &= \ln \left( \sum_{j=1}^n X(j) \exp\{ \alpha E_j \cdot CX \} \right) - \alpha \sum_{i \in \mathcal{C}(Y)} Y(i)  (E_i \cdot CX)\notag\\
                       &\geq \alpha  \sum_{j=1}^n X(j) (E_j \cdot CX) - \alpha \sum_{i \in \mathcal{C}(Y)} Y(i)  (E_i \cdot CX)\notag\\
                       &= \alpha (X \cdot CX - Y \cdot CX).\notag
\end{align}
}
Therefore, 
\begin{align*}
X \cdot CX - Y \cdot CX \geq 0 \Rightarrow \forall \alpha > 0: RE(Y, \hat{X}) - RE(Y, X) \geq 0.
\end{align*}
If $\max_{i \in \mathcal{C}(X)} \{ (CX)_i \} > \min_{i \in \mathcal{C}(X)} \{ (CX)_i \}$, since Jensen's inequality is strict,
\begin{align*}
X \cdot CX - Y \cdot CX \geq 0 \Rightarrow \forall \alpha > 0: RE(Y, \hat{X}) - RE(Y, X) > 0.
\end{align*}
This completes the proof.
\end{proof}

\begin{proof}[Second proof of Lemma \ref{cool_hedge_2}]
By the convexity lemma (Lemma \ref{convexity_lemma}), unless $X$ is a fixed point, $RE(Y, \hat{X})$ is strictly convex. \eqref{valentine} implies that
\begin{align*}
\left. \frac{d}{d \alpha} RE(Y, \hat{X}) \right|_{\alpha = 0} = (X - Y) \cdot CX \geq 0.
\end{align*}
Noting that $\left. RE(Y, \hat{X}) \right|_{\alpha = 0} = RE(Y, X)$ completes the proof.
\end{proof}

\subsection{Divergence of iterated Hedge}

Let us now get to the proof that iterated Hedge diverges, in particular, in symmetric zero-sum games equipped with an interior equilibrium. Given an antisymmetric matrix $C$ equipped with an interior equilibrium, say $X^*$, it is simple to show that $X^*$ satisfies the relation
\begin{align*}
\forall X \in \mathbb{X}(C) : (X^* - X) \cdot CX = 0.
\end{align*}
As an example, consider the {\em rock-paper-scissors game}, which is a zero-sum symmetric bimatrix game with payoff matrix
\begin{align*}
C = \left(\begin{array}{ r  r  r }
0 & -1 & 1\\
1 & 0 & -1\\
-1 & 1 & 0\\
\end{array}\right).
\end{align*}
$C$ is anti-symmetric, that is, $C^T = -C$, therefore, for all $X$, $X \cdot CX = 0$. Furthermore, $X^* = (1/3, 1/3, 1/3)$ is the unique equilibrium strategy, implying after straight algebra that, for all $X \in \mathbb{X}(C)$, $(X^* - X) \cdot CX = 0$. Lemma \ref{cool_hedge_2} implies that starting anywhere in the interior of $\mathbb{X}(C)$ other than the uniform strategy (which is the equilibrium strategy), under any sequence of positive learning rates, the relative entropy distance between $X^*$ and $T^k(X^0)$ diverges to $\infty$ as $k \rightarrow \infty$.

\subsection{Precise bounds on equilibrium approximation}

Let us now prove that, in sharp contrast, averaging iterated Hedge converges. In fact, we give an equilibrium fully polynomial time approximation scheme: Given any desired equilibrium approximation error $\epsilon$, we compute a fixed learning rate $\alpha$ such that the average of iterated Hedge converges to an $\epsilon$-approximate Nash equilibrium of the corresponding symmetric zero-sum game.

\begin{lemma}
\label{convexity_lemma_normalized}
Let $C \in \mathbb{\hat{C}}$. Then, for all $Y \in \mathbb{X}(C)$ and for all $X \in \mathbb{\mathring{X}}(C)$, we have that
\begin{align*}
\forall \alpha > 0 : RE(Y, T(X)) \leq RE(Y, X) - \alpha (Y-X) \cdot CX + \alpha (\exp\{\alpha\} - 1).
\end{align*}
\end{lemma}

\begin{proof}
Using the assumption $C \in \mathbb{\hat{C}}$, we have
\begin{align*}
\sum_{j=1}^n X(j) (CX)_j^2 \leq \sum_{j=1}^n X(j) (CX)_j \leq 1.
\end{align*}
Therefore,
\begin{align*}
\frac{d}{d \alpha} RE(Y, T(X)) \leq X \cdot CX - Y \cdot CX + (\exp\{ \alpha \} - 1).
\end{align*}
Combining with \eqref{ooone} yields the lemma.
\end{proof}

\begin{lemma}
\label{approximation_lemma_corollary}
Let $C \in \mathbb{\hat{C}}$, $X^k \equiv T^k(X^0)$, $X^0 \in \mathbb{\mathring{X}}(C)$, and assume $\alpha > 0$ is held constant. Then,
\begin{align}
\forall \theta > 0 \mbox{ } \forall Y \in \mathbb{X}(C) : \frac{1}{K+1} \sum_{k=0}^K (Y - X^k) \cdot CX^k \leq (\exp\{\alpha\}-1) + \theta\label{limeisgood}
\end{align}
where $K \geq \floor*{RE(Y, X^0)/(\alpha \theta)}$. If $X^0$ is the uniform strategy, then \eqref{limeisgood} holds after $\floor*{\ln(n) / (\alpha \theta)}$ iterations and continues to hold thereafter.
\end{lemma}

\begin{proof}
Assume for the sake of contradiction that \eqref{limeisgood} does not hold, that is,
\begin{align}
\frac{1}{K+1} \sum_{k=0}^K (Y - X^k) \cdot CX^k > (\exp\{\alpha\}-1) + \theta.\label{red2}
\end{align}
Invoking Lemma \ref{convexity_lemma_normalized},
\begin{align*}
RE(Y, X^{k+1}) \leq RE(Y, X^k) - \alpha (Y-X^k) \cdot CX^k + \alpha (\exp\{\alpha\} - 1).
\end{align*}
Summing over $k = 0, \ldots, K$, we obtain
\begin{align*}
RE(Y, X^{K+1}) \leq RE(Y, X^0) - \alpha \sum_{k=0}^K (Y-X^k) \cdot CX^k + (K+1) \alpha (\exp\{\alpha\} - 1)
\end{align*}
and, therefore,
\begin{align}
\frac{RE(Y, X^{K+1})}{K+1} \leq \frac{RE(Y, X^0)}{K+1} - \frac{\alpha}{K+1} \sum_{k=0}^K (Y-X^k) \cdot CX^k + \alpha (\exp\{\alpha\} - 1).\label{lalala2}
\end{align}
Substituting then \eqref{red2} in \eqref{lalala2} we obtain
\begin{align*}
\frac{RE(Y, X^{K+1})}{K+1} \leq \frac{RE(Y, X^0)}{K+1} - \alpha \theta,
\end{align*}
which implies that
\begin{align*}
RE(Y, X^{K+1}) \leq RE(Y, X^0) - (K+1) \alpha \theta
\end{align*}
and, therefore, that
\begin{align*}
RE(Y, X^0) \geq (K+1) \alpha \theta.
\end{align*}
But this contradicts the previous definition of $K$ and completes the proof.\\

The second part of the lemma is implied from the observation that a convex function is maximized at the boundary and, in our particular case, the vertices of the probability simplex.
\end{proof}

\begin{theorem}
\label{fptas_theorem}
Let $C \in \mathbb{\hat{C}}$ be such that it has been obtained by an affine transformation on a antisymmetric matrix. Then starting at the uniform strategy, the average of iterated Hedge converges to an $\epsilon$-approximate symmetric Nash equilibrium strategy in at most
\begin{align*}
\floor*{\frac{\ln(n)}{\frac{\epsilon}{2}\ln\left(1+ \frac{\epsilon}{2} \right)}}
\end{align*} 
iterations using a fixed learning rate equal to $\ln(1+\epsilon/2)$.
\end{theorem}

\begin{proof}
This theorem is a simple implication of Lemma \ref{approximation_lemma_corollary}. Note that since $Y$ is arbitrary in \eqref{limeisgood}, we may write it as
\begin{align*}
\max_{i=1}^n \left\{ C \left( \frac{1}{K+1} \sum_{k=0}^K X^k \right) \right\} - \frac{1}{K+1} \sum_{k=0}^K X^k \cdot CX^k \leq (\exp\{\alpha\} - 1) + \theta.
\end{align*}
Using the notation
\begin{align*}
\bar{X}^K \equiv \frac{1}{K+1} \sum_{k=0}^K X^k
\end{align*}
and using also the assumption that $C$ has been obtained by an affine transformation on a antisymmetric matrix, we obtain that
\begin{align*}
(C\bar{X}^K)_{\max} - \bar{X}^K \cdot C\bar{X}^K \leq (\exp\{\alpha\} - 1) + \theta.
\end{align*}
Letting $\theta = \epsilon/2$ and $\alpha = \ln(1+\epsilon/2)$ and applying Lemma \ref{approximation_lemma_corollary}, we obtain the theorem.
\end{proof}

We note that the previous analysis nearly exactly matches the bound in \citep[Section 6.1]{FreundSchapire2} although these bounds have been obtained using different analytical routes.

\section{A variational perspective on algorithmic boosting}
\label{ReplicatorBoosting}

Our definition of algorithmic boosting of discrete maps extends in a natural manner to continuous flows. In this section, we consider algorithmically boosting the replicator dynamic, which is given by the following differential equation:
\begin{align*}
\dot{X}(i) = X(i) \left( (CX)_i - X \cdot CX \right), \quad i = 1, \ldots, n.
\end{align*} 
In fact, from the perspective of computing Nash equilibria in symmetric bimatrix games (and, more generally, solving {\em variational inequalities} over the standard simplex), Hedge can be meaningfully understood as a discretization of the replicator dynamic. Our first task in this section is to prove this duality between Hedge and the replicator dynamic. Then, as our main result in this section, we prove that the $\omega$-limit set of the long-time average of the replicator dynamic in a symmetric zero-sum game consists entirely of Nash equilibria in this game. In this result, we observe a phenomenon that is analogous to that of the previous section, namely, that although an orbit may not in itself converge to the desired fixed point, by algorithmically boosting the orbit we obtain convergence.

\subsection{Hedge is a discretization of the replicator dynamic}

In this part of this section, we assume that $C$ is an, in general, nonlinear operator. Let us first note that {\em consistency} and {\em convergence} are standard properties numerical integrators satisfy (for example, see \citep{Burden}). That Hedge is a consistent numerical integrator for the replicator dynamic rests on the observation that
\begin{align*}
\frac{d}{d \alpha} \left. \left( X(i) \frac{\exp \left\{ \alpha (CX)_i \right\}}{\sum_{j=1}^n X(j) \exp\{\alpha (CX)_j \}} \right) \right|_{\alpha = 0} = X(i) \left( (CX)_i - X \cdot CX \right), i = 1, \ldots, n.
\end{align*}
We note that under a stochastic model of evolution, a similar observation has been leveraged for the study of dynamics in {\em congestion games} by \cite{Piliouras}. Using the previous observation, we also prove convergence by comparing the error Hedge generates relative to the Euler method, which approximates the replicator dynamic using iterates generated by the difference equation
\begin{align*}
W^{k+1}(i) = W^k(i) + \alpha W(i) \left( (CW)_i - W \cdot CW \right), k=0, \ldots, K.
\end{align*}
Starting from the interior of the simplex, for any finite $K$, there exists $\alpha$ such that Euler's method remains in the interior. Therefore, for small enough time step, Euler's method remains well-defined given any number of finite iterations. Euler's method is convergent under the assumption that the replicator equation is Lipschitz and under the assumption that the second derivative of the solution trajectory with respect to time is bounded. We have the following theorem:

\begin{theorem}
Under the aforementioned assumptions that ensure that the Euler method is a convergent numerical integrator for the replicator dynamic and under the further assumption that $\{\max\{ \| CX \| | X \in \mathbb{X}(C)\} < \infty$, Hedge is a convergent numerical integrator for the replicator dynamic.
\end{theorem}

\begin{proof}
The Taylor expansion of $T$ at $\alpha = 0$ gives
\begin{align*}
T_i(X) = X(i) + \alpha X(i) ( E_i \cdot CX - X \cdot CX) + \frac{\alpha^2}{2} \left. \frac{d^2 T_i(X)}{d \alpha^2} \right|_{\alpha = \xi_i}.
\end{align*}
Using the notation
\begin{align*}
F_i(X) \equiv X(i) ( E_i \cdot CX - X \cdot CX),
\end{align*}
we obtain
\begin{align*}
T_i(X) = X(i) + \alpha F_i(X) + \frac{\alpha^2}{2} \left. \frac{d^2 T_i(X)}{d \alpha^2} \right|_{\alpha = \xi_i}.
\end{align*}
Let us assume $F$ is Lipschitz with constant $L$. Furthermore, under the assumption that $\{\max\{ \| CX \| | X \in \mathbb{X}(C)\} < \infty$, there exists a positive constant $M$ such that 
\begin{align*}
\forall i = 1, \ldots n \mbox{ }\forall X \in \mathbb{X}(C) : \left| \frac{d^2 T_i(X)}{d \alpha^2} \right| \leq M.
\end{align*}
To show convergence, note that Hedge gives
\begin{align*}
X^{k+1}(i) = X^k(i) + \alpha F_i(X^k) + \frac{\alpha^2}{2} \left. \frac{d^2 T_i(X)}{d \alpha^2} \right|_{\alpha = \xi^k_i}
\end{align*}
whereas Euler's method gives
\begin{align*}
W^{k+1}(i) = W^k(i) + \alpha F_i(W^k).
\end{align*}
Subtracting these equations, we obtain
\begin{align*}
X^{k+1}(i) - W^{k+1}(i) = X^k(i) - W^k(i) + \alpha (F_i(X^k) - F_i(W^k)) + \frac{\alpha^2}{2} \left. \frac{d^2 T_i(X)}{d \alpha^2} \right|_{\alpha = \xi_i^k}
\end{align*}
and hence
\begin{align*}
|X^{k+1}(i) - W^{k+1}(i)| \leq |X^k(i) - W^k(i)| + \alpha |F_i(X^k) - F_i(W^k)| + \frac{\alpha^2}{2} \left| \left. \frac{d^2 T_i(X)}{d \alpha^2} \right|_{\alpha = \xi_i^k} \right|
\end{align*}
Since, as noted above, $F$ is Lipschitz with parameter $L$ and 
\begin{align*}
\left| \left. \frac{d^2 T_i(X)}{d \alpha^2} \right|_{\alpha = \xi_k} \right| \leq M
\end{align*}
we obtain
\begin{align*}
|X^{k+1}(i) - W^{k+1}(i)| \leq (1 + \alpha L) |X^k(i) - W^k(i)| + \frac{\alpha^2 M}{2}.
\end{align*}
To proceed further, we need the following lemma:

\begin{lemma}
\label{bound_lemma}
If $s$ and $t$ are positive real numbers, $\{a_i\}_{i = 0}^k$ is a sequence satisfying $a_0 \geq -t/s$ and
\begin{align*}
a_{i+1} \leq (1+s)a_i + t, \forall i = 0, \ldots, k-1
\end{align*}
then
\begin{align*}
a_{i+1} \leq \exp\{(i+1)s\} \left( a_0 + \frac{t}{s} \right) - \frac{t}{s}.
\end{align*}
\end{lemma}

\begin{proof}
See \citep{Burden}.
\end{proof}

Applying Lemma \ref{bound_lemma}, we further obtain
\begin{align*}
|X^{k+1}(i) - W^{k+1}(i)| \leq \exp\{ (k+1) \alpha L \} \left( |X^0(i) - W^0(i)| + \frac{\alpha M}{2 L} \right) - \frac{\alpha M}{2L}
\end{align*}
which implies
\begin{align*}
|X^{k+1}(i) - W^{k+1}(i)| \leq (\exp\{ (k+1) \alpha L \} -1) \frac{\alpha M}{2 L}.
\end{align*}
Therefore, as $\alpha \rightarrow 0$, we have that $|X^{k+1}(i) - W^{k+1}(i)| \rightarrow 0$, and, given that Euler's method is convergent, Hedge is similarly convergent.
\end{proof}

\subsection{Convergence of the long-run average of the replicator dynamic}

\begin{theorem}
\label{variationalboostingtheorem}
Let $C$ be the payoff matrix of a symmetric zero-sum game $(C, C^T)$. Furthermore, let $X_\alpha, \alpha \in [0, t]$ be an orbit of the replicator dynamic 
\begin{align*}
\dot{X} = X(i) ( (CX)_i - X \cdot CX), i = 1, \ldots, n.
\end{align*}
Then the $\omega$-limit set of the long-run average
\begin{align*}
 \frac{1}{t} \int_{0}^{t} X_{\omega} d \omega
\end{align*}
consists entirely of symmetric Nash equilibrium strategies of $(C, C^T)$.
\end{theorem}

\begin{proof}
Let
\begin{align*}
\bar{X}_t = \frac{1}{t} \int_0^t X_\omega d\omega
\end{align*}
where $X_\omega, \omega \in [0, t]$ is a trajectory of the replicator dynamic. Furthermore, let $\{ \bar{X}_{t_n} \}_0^{\infty}$ be a convergent subsequence such that $\{ \bar{X}_{t_n} \} \rightarrow \bar{X}$. That is, $\bar{X}$ is in the $\omega$-limit set of the long-time average $\bar{X}_t$. We will show that $\bar{X}$ is a Nash equilibrium strategy. To that end, let
\begin{align*}
\left\{ \frac{1}{t_m} \int_0^{t_m} (X_\omega \cdot CX_\omega) d\omega \right\}_{m=0}^{\infty}
\end{align*}
be a convergent subsequence of the sequence of long-time averages. Note now that letting $Y \in \mathbb{X}(C)$ be arbitrary, since
\begin{align*}
\int_0^t \frac{d}{d \omega} RE(Y, X_\omega) d\omega = RE(Y, X_t) - RE(Y, X_0),
\end{align*}
we obtain that
\begin{align*}
\liminf\limits_{m \rightarrow \infty} \left\{ \frac{1}{t_m} \int_0^{t_m} \frac{d}{d \omega} RE(Y, X_\omega) d\omega \right\} \geq 0
\end{align*}
where we have used that the relative entropy function is nonnegative. Straight calculus (for example, see \citep[p. 98]{Weibull}) gives that
\begin{align*}
\frac{d}{d\omega} RE(Y, X_\omega) = X_\omega \cdot CX_\omega - Y \cdot CX_\omega.
\end{align*}
Combining the previous relations, we obtain
\begin{align}
\forall Y \in \mathbb{X}(C) : Y \cdot C\bar{X} \leq \lim_{m \rightarrow \infty} \left\{ \frac{1}{t_m} \int_0^{t_m} (X_\omega \cdot CX_\omega) d\omega \right\}.\label{one}
\end{align}
Now let $\{ X_{t_k} \}_{0}^{\infty}$ be a convergent subsequence of the sequence
\begin{align*}
\left\{ \frac{1}{t_m} \int_0^{t_m} (X_\omega \cdot CX_\omega) d\omega \right\}_{m=0}^{\infty}
\end{align*}
that converges to $X$. Furthermore, let $E_i$ be a pure strategy in the carrier of $X$. Then, since the relative entropy is bounded, we obtain that
\begin{align*}
\lim_{k \rightarrow \infty} \left\{ \frac{1}{t_k} \int_0^{t_k} \frac{d}{d \omega} RE(E_i, X_\omega) d\omega \right\} = 0,
\end{align*}
which implies that
\begin{align*}
(C\bar{X})_i = \lim_{m \rightarrow \infty} \left\{ \frac{1}{t_m} \int_0^{t_m} (X_\omega \cdot CX_\omega) d\omega \right\}.
\end{align*}
Therefore, for every pure strategy $E_i$ in the carrier of $X$, we have that
\begin{align*}
(C\bar{X})_i = \lim_{m \rightarrow \infty} \left\{ \frac{1}{t_m} \int_0^{t_m} (X_\omega \cdot CX_\omega) d\omega \right\}
\end{align*}
and, for every pure strategy $E_j$ outside the carrier of $X$, we have that
\begin{align*}
(C\bar{X})_j \leq \lim_{m \rightarrow \infty} \left\{ \frac{1}{t_m} \int_0^{t_m} (X_\omega \cdot CX_\omega) d\omega \right\}
\end{align*}
which implies that
\begin{align*}
(C\bar{X})_{\max} = \lim_{m \rightarrow \infty} \left\{ \frac{1}{t_m} \int_0^{t_m} (X_\omega \cdot CX_\omega) d\omega \right\}.
\end{align*}
Assuming now the game is zero-sum so that $C$ is an antisymmetric matrix (which implies that, for all $x$, $x \cdot Cx =0$, we obtain from the previous inequality that
\begin{align*}
(C\bar{X})_{\max} = 0 = \bar{X} \cdot C \bar{X},
\end{align*}
and, therefore, $\bar{X}$ is a Nash equilibrium as claimed.
\end{proof}

\section{Nash equilibrium computation using clairvoyant averaging}
\label{constructive}

Given the payoff matrix $C$ of a symmetric bimatrix game $(C, C^T)$, recall that Hedge is given by map $\mathsf{T} : \mathbb{X}(C) \rightarrow \mathbb{X}(C)$ where
\begin{align*}
\mathsf{T}_i(X) = X(i) \frac{\exp\{\alpha (CX)_i \}}{\sum_{j=1}^n X(j) \exp\{\alpha (CX)_j\}}, i =1, \ldots, n,
\end{align*}
and $\alpha$ is a parameter called the learning rate. In this paper, we consider the map $\mathsf{T} : \mathbb{X}(C) \times \mathbb{X}(C) \rightarrow \mathbb{X}(C)$ where
\begin{align*}
\mathsf{T}_i(X|Z) = X(i) \frac{\exp\{\alpha (CZ)_i \}}{\sum_{j=1}^n X(j) \exp\{\alpha (CZ)_j\}}, i =1, \ldots, n,
\end{align*}
and $Z \in \mathbb{X}(C)$. If $\alpha$ is small, the previous map can be approximated by the map $T : \mathbb{X}(C) \times \mathbb{X}(C) \rightarrow \mathbb{X}(C)$, where
\begin{align*}
T_i(X|Z) = X(i) + \alpha X(i) ((CZ)_i - X \cdot CZ), i =1, \ldots, n.
\end{align*}
This is the Euler approximation. In this section, we will work exclusively with this approximation.

\subsection{A formula for the equilibrium approximation error}

Let us assume $0 \leq C \leq 1$. Furthermore, let $T(X|Z) \equiv \hat{X}$. Straight algebra gives
\begin{align*}
\frac{\hat{X}(i)}{\hat{X}(j)} = \frac{X(i)}{X(j)} \frac{1 + \alpha ((CZ)_i - X \cdot CZ)}{1 + \alpha ((CZ)_j - X \cdot CZ)}
\end{align*}
and taking logarithms on both sides we obtain
\begin{align*}
\ln \left( \frac{\hat{X}(i)}{\hat{X}(j)} \right) = \ln \left( \frac{X(i)}{X(j)} \right) + \ln \left( \frac{1 + \alpha ((CZ)_i - X \cdot CZ)}{1 + \alpha ((CZ)_j - X \cdot CZ)} \right)
\end{align*}
which implies that
\begin{align*}
\ln \left( \frac{\hat{X}(i)}{\hat{X}(j)} \right) = \ln \left( \frac{X(i)}{X(j)} \right) + \ln \left( 1 + \alpha ((CZ)_i - X \cdot CZ) \right) - \ln \left( 1 + \alpha ((CZ)_j - X \cdot CZ) \right)
\end{align*}
which further implies using standard inequalities involving the logarithm
\begin{align*}
\ln \left( \frac{\hat{X}(i)}{\hat{X}(j)} \right) \geq \ln \left( \frac{X(i)}{X(j)} \right) + \frac{\alpha ((CZ)_i - X \cdot CZ)}{1 + \alpha ((CZ)_i - X \cdot CZ)} - \alpha ((CZ)_j - X \cdot CZ)
\end{align*}
which even further implies using the Taylor expansion of the reciprocal and using the assumption $0 \leq C \leq 1$ that
\begin{align*}
O(\alpha^3) + \alpha^2 + \ln \left( \frac{\hat{X}(i)}{\hat{X}(j)} \right) \geq \ln \left( \frac{X(i)}{X(j)} \right) + \alpha ((CZ)_i - X \cdot CZ) - \alpha ((CZ)_j - X \cdot CZ)
\end{align*}
which even further implies using straight algebra that
\begin{align*}
O(\alpha^3) + \alpha^2 + \ln \left( \frac{\hat{X}(i)}{\hat{X}(j)} \right) \geq \ln \left( \frac{X(i)}{X(j)} \right) + \alpha ((CZ)_i - (CZ)_j).
\end{align*}
We may write the previous equation as
\begin{align*}
O(\alpha_k^3) + \alpha_k^2 + \ln \left( \frac{X^{k+1}(i)}{X^{k+1}(j)} \right) \geq \ln \left( \frac{X^k(i)}{X^k(j)} \right) + \alpha_k ((CZ^k)_i - (CZ^k)_j).
\end{align*}
Denoting $\bar{\alpha}$ an upper bound on the sequence of learning rates used by the algorithm, summing from $k=0$ to $k=K$ and dividing by $A_K$ and rearranging, we obtain
\begin{align*}
O(\bar{\alpha}^2) + \bar{\alpha} + \frac{1}{A_K} \ln \left( \frac{X^{K+1}(i)}{X^{K+1}(j)} \right) \geq \frac{1}{A_K} \ln \left( \frac{X^0(i)}{X^0(j)} \right) + ((C\bar{Z}^K)_i - (C\bar{Z}^K)_j)
\end{align*}
which implies
\begin{align*}
O(\bar{\alpha}^2) + \bar{\alpha} + \frac{1}{A_K} \ln \left( \frac{X^{K+1}(i)}{X^{K+1}(j)} \right) \geq \frac{1}{A_K} \ln \left( \frac{\min_iX^0(i)}{ \max_j X^0(j)} \right) + ((C\bar{Z}^K)_i - (C\bar{Z}^K)_j)
\end{align*}
which further implies
\begin{align*}
O(\bar{\alpha}^2) + \bar{\alpha} + \frac{1}{A_K} &\sum_{j=1}^n \bar{Z}^K(j) \ln \left( \frac{X^{K+1}(i_{\max})}{X^{K+1}(j)} \right) \geq\\
& \geq \frac{1}{A_K} \ln \left( \frac{\min_iX^0(i)}{ \max_j X^0(j)} \right) + ((C\bar{Z}^K)_{\max} - \bar{Z}^K \cdot C\bar{Z}^K)
\end{align*}
which even further implies
\begin{align}
O(\bar{\alpha}^2) + \bar{\alpha} + \frac{1}{A_K} &\sum_{j=1}^n \left( \frac{1}{A_K} \sum_{k=0}^K \alpha_k Z^k(j) \right) \ln \left( \frac{X^{K+1}(i_{\max})}{X^{K+1}(j)} \right) \geq\notag\\
& \geq \frac{1}{A_K} \ln \left( \frac{\min_iX^0(i)}{ \max_j X^0(j)} \right)+ ((C\bar{Z}^K)_{\max} - \bar{Z}^K \cdot C\bar{Z}^K)\notag
\end{align}
which even further implies
\begin{align}
O(\bar{\alpha}^2) + \bar{\alpha} - \frac{1}{A_K} &\sum_{j=1}^n \left( \frac{1}{A_K} \sum_{k=0}^K \alpha_k Z^k(j) \right) \ln \left( X^{K+1}(j) \right) \geq\notag\\
& \geq \frac{1}{A_K} \ln \left( \frac{\min_iX^0(i)}{ \max_j X^0(j)} \right) + ((C\bar{Z}^K)_{\max} - \bar{Z}^K \cdot C\bar{Z}^K).\label{super}
\end{align}

\subsection{Our algorithm and one of its fundamental properties}

Our algorithm starts at an interior strategy $X^0 \in \mathbb{\mathring{X}}(C)$ and at every iteration $K = 0, 1, \ldots$ it computes $\bar{\alpha} \geq \alpha_K > 0$ and $Z^K \in \mathbb{\mathring{X}}(C)$ as a solution of the convex optimization problem
\begin{align*}
\mbox{minimize} \quad &RE(X^K, Z)\\
\mbox{subject to} \quad &Z(i) = X^K(i) + A_K X^K(i) ((CZ)_i - X^K \cdot CZ), i =1, \ldots, n\\
&Z \in \mathbb{X}(C)
\end{align*}
where $RE(\cdot, \cdot)$ is the relative entropy function,
\begin{align*}
A_K = \sum_{k=0}^K \alpha_k,
\end{align*}
and $\alpha_K$ is chosen small enough such that
\begin{align*}
\min\left\{ X^{K+1}(i) + A_K X^{K+1}(i) ((C\mathcal{Y})_i - X^{K+1} \cdot C\mathcal{Y}) \bigg| \mathcal{Y} \in \mathbb{X}(C) \right\} > 0, i =1,\ldots,n,
\end{align*}
where
\begin{align*}
X^{K+1} = \frac{1}{A_K} \sum_{k=0}^K \alpha_k Z^k.
\end{align*}
An $\bar{\alpha} \geq \alpha_K > 0$ satisfying these constraints is shown to always exist. A small enough $\alpha_K$ can be obtained by a halving schema. We first need two elementary lemmas:

\begin{lemma}
\label{fundamental}
$\forall X \in \mathbb{\mathring{X}}(C)$ and $\forall Z \in \mathbb{X}(C)$, there exists $\hat{\alpha} > 0$ (which may depend on $Z$) such that $\forall 0\leq \alpha < \hat{\alpha}$ we have that 
\begin{align*}
X(i) + \alpha X(i) ((CZ)_i - X \cdot CZ) \in \mathbb{\mathring{X}}(C), i =1,\ldots,n. 
\end{align*}
Since $\mathbb{X}(C)$ is compact, this implies that there exists $\hat{\alpha} > 0$ such that $\forall 0\leq \alpha < \hat{\alpha}$ and $\forall Z \in \mathbb{X}(C)$ we have that 
\begin{align*}
X(i) + \alpha X(i) ((CZ)_i - X \cdot CZ) \in \mathbb{\mathring{X}}(C), i =1,\ldots,n.
\end{align*}
\end{lemma}

\begin{proof}
Since $X$ is an interior probability vector and since, for all $\alpha > 0$ and for all $Z \in \mathbb{X}(C)$, we have that
\begin{align*}
\sum_{i=1}^n (X(i) + \alpha X(i) ((CZ)_i - X \cdot CZ)) = 1
\end{align*}
the vector whose elements are
\begin{align*}
X(i) + \alpha X(i) ((CZ)_i - X \cdot CZ), i =1,\ldots,n 
\end{align*}
remains on the tangent space of the simplex $\mathbb{X}(C)$. This implies that for all $Z \in \mathbb{X}(C)$, there exists $\hat{\alpha}_Z > 0$ such that $\forall 0\leq \alpha < \hat{\alpha}_Z$ we have that 
\begin{align*}
X(i) + \alpha X(i) ((CZ)_i - X \cdot CZ) \in \mathbb{\mathring{X}}(C), i =1,\ldots,n
\end{align*}
as claimed. The second part of the lemma follows by the compactness of $\mathbb{X}(C)$ and the continuity of the maximum $\hat{\alpha}_Z$ as a function of $Z \in \mathbb{X}(C)$ which together imply that
\begin{align*}
\min \left\{ \hat{\alpha}_Z | Z \in \mathbb{X}(C) \right\} > 0.
\end{align*}
This completes the proof.
\end{proof}

\begin{lemma}
\label{fundamentaal}
Let $X \in \mathbb{\mathring{X}}(C)$ be such that, for all $Z \in \mathbb{X}(C)$, we have that
\begin{align*}
X(i) + A X(i) ((CZ)_i - X \cdot CZ) \in \mathbb{\mathring{X}}(C), i =1,\ldots,n
\end{align*}
for some $A > 0$. Then there exists $\hat{\alpha} > 0$ such that $\forall 0\leq \alpha < \hat{\alpha}$ and $\forall Z \in \mathbb{X}(C)$ we have that 
\begin{align*}
X(i) + (A + \alpha) X(i) ((CZ)_i - X \cdot CZ) \in \mathbb{\mathring{X}}(C), i =1,\ldots,n.
\end{align*}
\end{lemma}

\begin{proof}
Granted that, as shown in the previous lemma, for all $\alpha > 0$ the vector whose elements are
\begin{align*}
X(i) + (A + \alpha) X(i) ((CZ)_i - X \cdot CZ), i =1,\ldots,n
\end{align*}
is on the tangent space of the simplex $\mathbb{X}(C)$, this lemma is a simple implication of the intermediate value theorem.
\end{proof}

We have the following fundamental property:

\begin{lemma}
\label{fundamental2}
$\forall K \geq 0$, there exist  $\bar{\alpha} \geq \alpha_K > 0$ and $Z^K \in \mathbb{\mathring{X}}(C)$ satisfying the algorithm's constraints.
\end{lemma}

\begin{proof}
Our proof is by induction. Let us first prove the basis of the induction. We would like to show that there exists $a_0 > 0$ such that, for all $0 < \alpha_0 \leq a_0$, every solution $Z^0$ of the convex optimization problem
\begin{align*}
\mbox{minimize} \quad &RE(X^0, Z)\tag{$\mbox{OPT}_0$}\\
\mbox{subject to} \quad &Z(i) = X^0(i) + \alpha_0 X^0(i) ((CZ)_i - X^0 \cdot CZ), i =1, \ldots, n,\\
&Z \in \mathbb{X}(C)
\end{align*}
satisfies
\begin{align*}
\min\left\{ Z^0(i) + \alpha_0 Z^0(i) ((C\mathcal{Y})_i - Z^0 \cdot C\mathcal{Y}) \bigg| \mathcal{Y} \in \mathbb{X}(C) \right\} > 0, i =1,\ldots,n.\tag{$\mbox{OPT}^i_0$}
\end{align*}
Note that by Lemma \ref{fundamental} and Brouwer's fixed point theorem, there exists $\eta > 0$ such that, for all $\alpha_0 \leq \eta$, the upper optimization problem ($\mbox{OPT}_0$) is feasible. Note further that, by Berge's maximum principle, given the strict convexity of the objective function, the solutions of the upper optimization problem are a continuous function of $\alpha_0$. If $\alpha_0 = 0$, the solution of the upper optimization problem is interior by the assumption that $X^0$ is interior. Therefore, there exists $\eta' > 0$ such that, for all $0 < \alpha_0 \leq \eta'$, the solutions of the upper optimization problem are interior. Note now that, invoking again Berge's maximum principle, the objective functions of the lower optimization problems ($\mbox{OPT}^i_0$) are a continuous function of $Z^0$, which is itself a continuous function function of $\alpha_0$ and that when $\alpha_0 = 0$ the objective value is strictly positive (since $Z^0 = X^0$). Therefore, by the intermediate value theorem there exists $a_0 > 0$ such that, for all $0 < \alpha_0 \leq a_0$, the upper and lower optimization problems can be simultaneously satisfied as desired. For the induction step, we assume that
\begin{align*}
\min \left\{ X^{K}(i) + A_{K-1} X^{K}(i) ((C\mathcal{Y})_i - X^{K} \cdot C\mathcal{Y}) \bigg| \mathcal{Y} \in \mathbb{X}(C) \right\} > 0
\end{align*}
and we would like to show that there exists $\bar{\alpha} \geq \alpha_K > 0$ and $Z^K \in \mathbb{\mathring{X}}(C)$ such that
\begin{align*}
\mbox{minimize} \quad &RE(X^K, Z)\tag{$\mbox{OPT}_K$}\\
\mbox{subject to} \quad &Z(i) = X^K(i) + A_K X^K(i) ((CZ)_i - X^K \cdot CZ), i =1, \ldots, n\\
&Z \in \mathbb{X}(C)
\end{align*}
where $\alpha_K$ is chosen small enough such that
\begin{align*}
\min\left\{ X^{K+1}(i) + A_K X^{K+1}(i) ((C\mathcal{Y})_i - X^{K+1} \cdot C\mathcal{Y}) \bigg| \mathcal{Y} \in \mathbb{X}(C) \right\} > 0, i =1,\ldots,n\tag{$\mbox{OPT}^i_K$}
\end{align*}
The proof this pair of constraints admits a solution by extending $A_{K-1}$ to a larger value $A_K = A_{K-1} + \alpha_K$ follows by an argument analogous to that used in the induction basis: Observe that when $\alpha_K = 0$, the upper optimization problem ($\mbox{OPT}_K$) is feasible by the induction hypothesis and Brouwer's fixed point theorem. The induction hypothesis further guarantees that the feasible solution is interior. Lemma \ref{fundamentaal} and Brouwer's fixed point theorem further guarantee an interior solution as $\alpha_K$ increases from zero. Furthermore, Berge's maximum theorem ensures that the solutions are a continuous function of $\alpha_K$. Note now that, invoking again Berge's maximum principle, the objective functions of the lower optimization problems ($\mbox{OPT}^i_K$) are a continuous function of $Z^K$, which is itself a continuous function function of $\alpha_K$ and that when $\alpha_K = 0$ the objective values of the lower problems are strictly positive by the induction hypothesis. Therefore, by the intermediate value theorem there exists $a_K > 0$ such that, for all $0 < \alpha_K \leq a_K$, the upper and lower optimization problems can be simultaneously satisfied as desired. This completes the proof.
\end{proof}

\subsection{Our constructive proof of existence of the Nash equilibrium}

Lemma \ref{fundamental2} implies that on every iteration $K = 0, 1, \ldots$ we have that
\begin{align*}
Z^K(i) = X^K(i) + A_K X^K(i) ((CZ^K)_i - X^K \cdot CZ^K), i =1, \ldots, n,
\end{align*}
which further implies that
\begin{align*}
\frac{\alpha_K}{A_K} Z^K(i) = \frac{\alpha_K}{A_K}  X^K(i) + \alpha_K X^K(i) ((CZ^K)_i - X^K \cdot CZ^K), i =1, \ldots, n,
\end{align*}
which even further implies that
\begin{align*}
\frac{\alpha_K}{A_K} Z^K(i) = \left( 1 - \frac{A_{K-1}}{A_{K-1} + \alpha_K} \right) X^K(i) + \alpha_K X^K(i) ((CZ^K)_i - X^K \cdot CZ^K), i =1, \ldots, n,
\end{align*}
which even further implies that
\begin{align*}
\frac{A_{K-1}}{A_{K-1} + \alpha_K} X^K(i) + \frac{\alpha_K}{A_K} Z^K(i) = X^K(i) + \alpha_K X^K(i) ((CZ^K)_i - X^K \cdot CZ^K), i =1, \ldots, n,
\end{align*}
which even further implies that
\begin{align*}
\frac{A_{K-1}}{A_K} X^K(i) + \frac{\alpha_K}{A_K} Z^K(i) = X^K(i) + \alpha_K X^K(i) ((CZ^K)_i - X^K \cdot CZ^K), i =1, \ldots, n,
\end{align*}
which, finally, implies that
\begin{align*}
X^{K+1}(i) \equiv \frac{1}{A_K} \sum_{k=0}^K \alpha_k Z^k(i) = X^K(i) + \alpha_K X^K(i) ((CZ^K)_i - X^K \cdot CZ^K), i =1, \ldots, n.
\end{align*}
Substituting in \eqref{super}, we obtain that
\begin{align*}
(C\bar{Z}^K)_{\max} - \bar{Z}^K \cdot C\bar{Z}^K \leq O(\bar{\alpha}^2) + \bar{\alpha} + \frac{1}{A_K} \ln \left( \frac{\max_iX^0(i)}{ \min_j X^0(j)} \right) - \frac{1}{A_K} \sum_{j=1}^n X^{K+1}(j) \ln \left( X^{K+1}(j) \right)
\end{align*}
and Jensen's inequality implies that
\begin{align*}
(C\bar{Z}^K)_{\max} - \bar{Z}^K \cdot C\bar{Z}^K \leq O(\bar{\alpha}^2) + \bar{\alpha} + \frac{1}{A_K} \ln \left( \frac{\max_iX^0(i)}{ \min_j X^0(j)} \right) + \frac{\ln(n)}{A_K}.
\end{align*}
This inequality implies that as $K \rightarrow \infty$, every limit point of the sequence $\{\bar{Z}^K\}$ is an $(O(\bar{\alpha}^2) + \bar{\alpha})$-approximate Nash equilibrium. Since $\bar{\alpha}$ can be made arbitrarily small, our proof is complete.

\section{The importance of a correct implementation of exponentiation}
\label{NumericalHedge}

A perspective on algorithmic boosting is that it generalizes the operation of taking the long-run average of an orbit. We have insofar focused on symmetric zero-sum games, where the long-run average of iterated Hedge had been known to converge to an approximate Nash equilibrium prior to our results. In this section, we focus on a symmetric bimatrix game that is not zero-sum, wherein the long-run average of iterated Hedge can, in principle, diverge. It is an open question if this is indeed the case. In this section, our main contribution is a numerical phenomenon whereby using the standard implementation of the exponential function in the computation of iterated Hedge, the average diverges whereas using an accurate implementation of the exponential function (in a fashion customized to this particular setting of computing Nash equilibria) the average converges. The example where we document this phenomenon is a great environment for testing iterative algorithms for computing Nash equilibria using the principles of algorithmic boosting.

\subsection{The Shapley game}

To the extent of our knowledge, the first study of the average of iterated Hedge outside the realm of zero-sum environments was by \cite{Daskalakis-SAGT}. \cite[Theorem 1]{Daskalakis-SAGT} shows divergence of the average of iterated Hedge playing against iterated Hedge (they consider a $2$-player setting) in Shapley's $3 \times 3$ symmetric bimatrix game whose payoff matrix is
\begin{align*}
C = \left[ \begin{array}{ccc}
0 & 1 & 2 \\
2 & 0 & 1 \\
1 & 2 & 0 \\
\end{array} \right]
\end{align*}
and whose unique Nash equilibrium is the symmetric Nash equilibrium corresponding to the uniform strategy $(1/3, 1/3, 1/3)$. Their result casts doubt that learning algorithms (used as fixed point iterators) can compute Nash equilibria. In this section (and broadly in this paper), we cast hope.

\subsection{The symmetric Shapley game}

In the rest of this section, we are concerned with the following symmetrization of Shapley's game:
\begin{align*}
C = \left[ \begin{array}{cccccc}
0 & 0 & 0 & 0 & 1 & 2 \\
0 & 0 & 0 & 2 & 0 & 1 \\
0 & 0 & 0 & 1 & 2 & 0 \\
0 & 1 & 2 & 0 & 0 & 0 \\
2 & 0 & 1 & 0 & 0 & 0 \\
1 & 2 & 0 & 0 & 0 & 0 \\
\end{array} \right],
\end{align*}
where the unique Nash equilibrium is the symmetric Nash equilibrium corresponding to the uniform strategy. In our numerical experiment, which is carried out in Matlab, we compare two different implementations of Hedge. The first implementation uses the aforementioned algebraic expression for generating iterates where the exponential function is implemented by Matlab's \texttt{exp(  )} routine. The second implementation is based on a formulation of Hedge as the solution of a convex optimization problem, in particular, as \cite{Krichene} point out
\begin{align*}
T(X) = \argmin \left\{ RE(Y, X) - \alpha Y \cdot CX \big| Y \in \mathbb{X}(C) \right\},
\end{align*}
where $RE(Y, X)$ is the relative entropy distance (Kullback-Leibler divergence) between probability vectors $Y$ and $X$ (as defined earlier). Note that this optimization problem is a {\em relative entropy program} \citep{Chandra} that admits a polynomial-time {\em interior point method} for its exact solution \citep{Nesterov}. In our experiment, we initialize both implementations with the probability vector $(0.1, 0.2, 0.3, 0.2, 0.1, 0.1)$ and use a learning rate equal to $10$. Figure \ref{FOCSHedge} plots the relative entropy distance between the long-run average of the iterates under the two implementations and the Nash equilibrium strategy. It is clear in the figure that the standard implementation gives divergence whereas the robust implementation using relative entropy programming (implemented using Matlab's \texttt{fmincon} interior point solver) gives convergence.

\begin{figure}[tb]
\centering
\includegraphics[width=12cm]{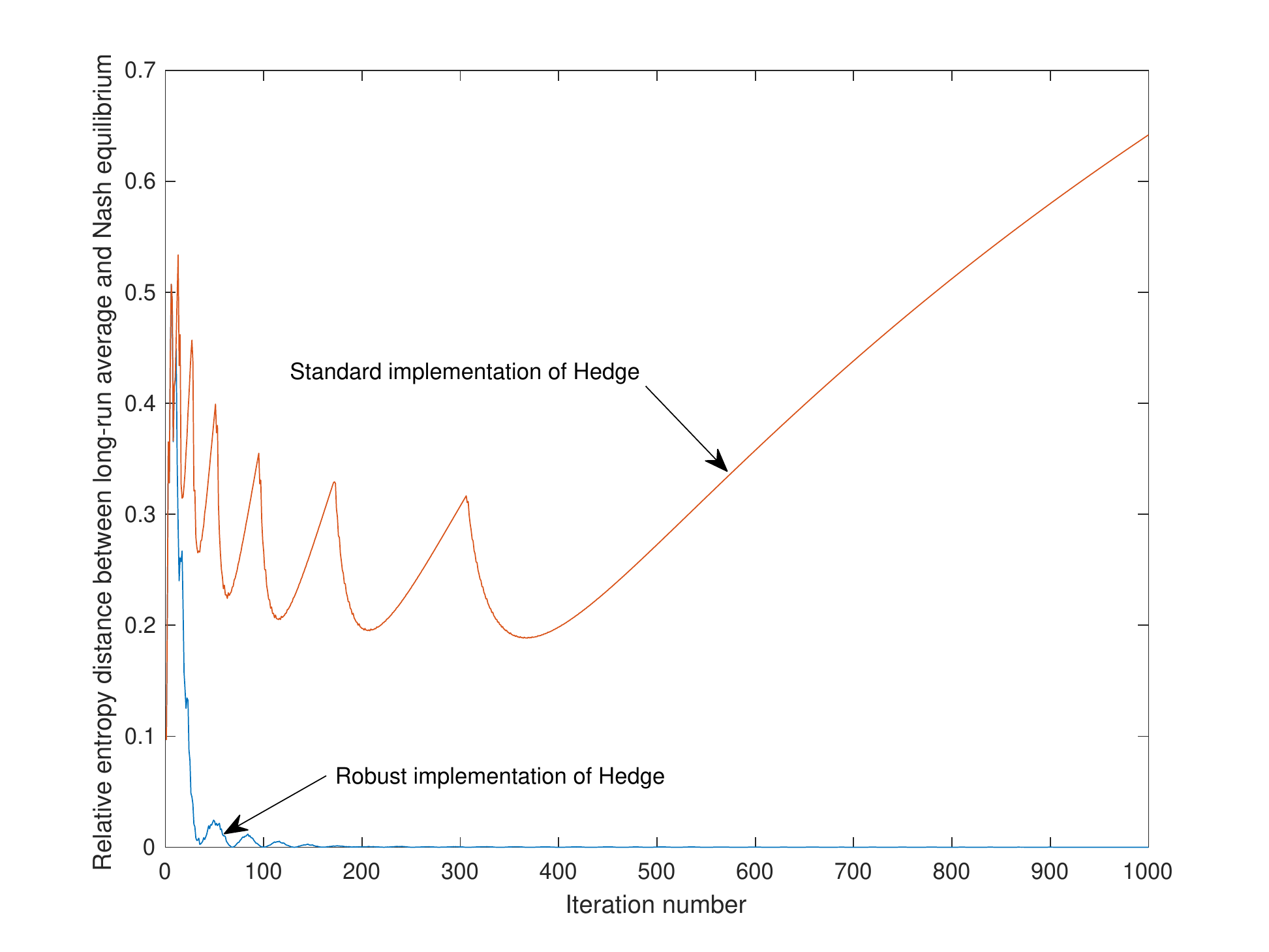}
\caption{\label{FOCSHedge}
Comparison of the standard implementation of Hedge and the robust implementation using relative entropy optimization. The learning rate is equal to $10$ and the initial condition is identical. The figure illustrates the large discrepancy between the two different implementations of the same algorithm.}
\end{figure}

\subsection{Discussion}

In closing this section, we would like to make two key observations. The first is that the long-run average of the replicator dynamic converges in the symmetric Shapley game. This is simple to check using Matlab's standard numerical integrator, namely, \texttt{ode45}. The second observation is that the principle of exponentiating and normalizing (for example, projecting onto the standard simplex) is used in a variety of machine learning tasks and in software code that is deployed in the field. Our experiment clearly demonstrates that it is not unlikely that implementations of the exponential function can trigger behavior different from what is expected or sought for simply because caution has not been paid to the correct implementation of exponentiation. Our hope is our experiment, and broadly this paper, squarely places the importance of correct implementations of exponentiation as a desideratum in field deployments of machine learning and fixed-point computation systems.

\section{A simplified numerically stable Hedge map}
\label{DoubleExponentiation}

In this section, we consider again \eqref{generalization}, which we repeat here for convenience:
\begin{align}
T_i(X | Z) = X(i) \frac{\exp\{\alpha (CZ)_i \}}{\sum_{j=1}^n X(j) \exp\{\alpha (CZ)_j\}}, i =1, \ldots, n.\label{generalization2}
\end{align}
Under this latter map, it is straightforward to show that the analogue of Lemma \ref{convexity_lemma_normalized} is:

\begin{lemma}
\label{convexity_lemma_normalized_2}
Let $C \in \mathbb{\hat{C}}$. Then, for all $Y \in \mathbb{X}(C)$ and for all $X \in \mathbb{\mathring{X}}(C)$, we have that
\begin{align*}
\forall \alpha > 0 : RE(Y, T(X|Z)) \leq RE(Y, X) - \alpha (Y-X) \cdot CZ + \alpha (\exp\{\alpha\} - 1).
\end{align*}
\end{lemma}

Using Lemma \ref{convexity_lemma_normalized_2} the analogue of Lemma \ref{approximation_lemma_corollary} is:

\begin{lemma}
\label{approximation_lemma_corollary_2}
Let $C \in \mathbb{\hat{C}}$, $\{ X^k \}$ be the sequence of iterates obtained by iteratively applying \eqref{generalization2} using a sequence $\{Z^k\}$ of multipliers, $X^0 \in \mathbb{\mathring{X}}(C)$, and assume $\alpha > 0$ is held constant. Then,
\begin{align}
\forall \theta > 0 \mbox{ } \forall Y \in \mathbb{X}(C) : \frac{1}{K+1} \sum_{k=0}^K (Y - X^k) \cdot CZ^k \leq (\exp\{\alpha\}-1) + \theta\label{limeisgood2}
\end{align}
where $K \geq \floor*{RE(Y, X^0)/(\alpha \theta)}$. If $X^0$ is the uniform strategy, then \eqref{limeisgood2} holds after $\floor*{\ln(n) / (\alpha \theta)}$ iterations and continues to hold thereafter.
\end{lemma}

We have the following theorem:

\begin{theorem}
\label{fptas_theorem_2}
Let $C \in \mathbb{\hat{C}}$ be such that it has been obtained by an affine transformation on a antisymmetric matrix. Then starting at the uniform strategy, the average of multipliers of iterated \eqref{generalization2} assuming
\begin{align*}
Z^k \in \arg\min \{ X^k \cdot C\mathcal{Y} | \mathcal{Y} \in \mathbb{X}(C) \}
\end{align*}
converges to an $\epsilon$-approximate symmetric Nash equilibrium strategy in at most
\begin{align*}
\floor*{\frac{\ln(n)}{\frac{\epsilon}{2}\ln\left(1+ \frac{\epsilon}{2} \right)}}
\end{align*} 
iterations using a fixed learning rate equal to $\ln(1+\epsilon/2)$.
\end{theorem}

\begin{proof}
This theorem is a simple implication of Lemma \ref{approximation_lemma_corollary_2}. Note that since $Y$ is arbitrary in \eqref{limeisgood2}, we may write it as
\begin{align*}
\max_{i=1}^n \left\{ C \left( \frac{1}{K+1} \sum_{k=0}^K Z^k \right) \right\} - \frac{1}{K+1} \sum_{k=0}^K X^k \cdot CZ^k \leq (\exp\{\alpha\} - 1) + \theta.
\end{align*}
Note now that
\begin{align*}
\frac{1}{K+1} \sum_{k=0}^K X^k \cdot CZ^k &= \frac{1}{K+1} \sum_{k=0}^K \min \left\{ X^k \cdot C\mathcal{Y} | \mathcal{Y} \in \mathbb{X}(C) \right\}\\
&\leq \min \left\{ \left( \frac{1}{K+1} \sum_{k=0}^K X^k \right) \cdot C\mathcal{Y} | \mathcal{Y} \in \mathbb{X}(C) \right\}\\
&\leq \left( \frac{1}{K+1} \sum_{k=0}^K X^k \right) \cdot C\left( \frac{1}{K+1} \sum_{k=0}^K X^k \right)
\end{align*}
and using the notations
\begin{align*}
\bar{X}^K \equiv \frac{1}{K+1} \sum_{k=0}^K X^k
\end{align*}
and
\begin{align*}
\bar{Z}^K \equiv \frac{1}{K+1} \sum_{k=0}^K Z^k
\end{align*}
and using also the assumption that $C$ has been obtained by an affine transformation on a antisymmetric matrix, which implies that
\begin{align*}
\left( \frac{1}{K+1} \sum_{k=0}^K X^k \right) \cdot C\left( \frac{1}{K+1} \sum_{k=0}^K X^k \right) = \bar{X}^K \cdot C\bar{X}^K = \bar{Z}^K \cdot C\bar{Z}^K
\end{align*}
we obtain that
\begin{align*}
(C\bar{Z}^K)_{\max} - \bar{Z}^K \cdot C\bar{Z}^K \leq (\exp\{\alpha\} - 1) + \theta.
\end{align*}
Letting $\theta = \epsilon/2$ and $\alpha = \ln(1+\epsilon/2)$ and applying Lemma \ref{approximation_lemma_corollary_2}, we obtain the theorem.
\end{proof}

The advantage of the extra optimization (which we show to be nearly effortless) is in numerical stability: Let $\mathring{C}$ be the antisymmetric matrix which $C$ has been obtained by an affine transformation from and note that
\begin{align*}
\min \{ X^k \cdot C\mathcal{Y} | \mathcal{Y} \in \mathbb{X}(C) \} = \min \{ X^k \cdot \mathring{C}\mathcal{Y} | \mathcal{Y} \in \mathbb{X}(C) \}.
\end{align*}
Note further that since $\mathring{C}$ is antisymmetric, we have that
\begin{align*}
\min \{ X^k \cdot \mathring{C}\mathcal{Y} | \mathcal{Y} \in \mathbb{X}(C) \} = - \max \{ \mathcal{Y} \cdot \mathring{C} X^k | \mathcal{Y} \in \mathbb{X}(C) \}.
\end{align*}
The latter optimization problem can be solved by looking up a pure strategy that maximizes $(\mathring{C}X^k)_i, i = 1,\ldots,n$ and, therefore, an optimizer can always be chosen to be a pure strategy. Therefore, we obtain that the multiplier $CZ^k, k = 0, 1, \ldots$ is always a column of matrix $C$, and, since there are $n$ columns, the exponentials can be computed ahead of the execution of the algorithm to any desired precision. The gain in numerical stability is obtained by looking up these precomputed accurate exponentials en route to the approximate Nash equilibrium our dynamic converges to.

\section{Discussion: Dominant eigenvectors and Nash equilibria}
\label{Pagerank-Nash}

Given a probability vector $X$, let $\Delta(X)$ denote the diagonal matrix whose diagonal elements are the elements of $X$. The following theorem relates dominant eigenvectors and Nash equilibria:

\begin{theorem}
Given a symmetric bimatrix game $(C, C^T)$ whose payoff matrix $C$ is nonnegative, $X^* \in \mathbb{X}(C)$ is a fixed point of the replicator dynamic (and, thus, of the Hedge map) if and only if there exists $\lambda > 0$ such that
\begin{align*}
\Delta(X^*) C X^* = \lambda X^*.
\end{align*}
Therefore, Nash equilibria are fixed points of this equation.
\end{theorem}

\begin{proof}
The forward direction is obvious and the reverse direction is obtained by left-multiplying with the inverse of the positive elements of $X^*$ in each vector position.
\end{proof}

In the previous theorem, it becomes clear that Nash equilibria are eigenvectors of a nonlinear operator. This formulation begs the question what the economic interpretation of the dominant eigenvector (PageRank) of a nonnegative matrix might be. It is certainly meaningful to ponder its fundamental property: Such dominant vector corresponds to a strategy whereby its payoff vector has the same ranking as the strategy itself in the sense that pure strategies that have a higher probability of being used receive a higher payoff. We leave the question of understanding the economic content of this phenomenon as future work. On the flip side it is also interesting to ponder what the Nash equilibria of the Google matrix correspond to in link analysis. In closing this discussion, we note that the formulation of Nash equilibria in the previous theorem as fixed points of a nonlinear operator yields an iterative algorithm intended to compute a Nash equilibrium as
\begin{align*}
\Delta(X^k) C X^{k+1} = \lambda_{k+1} X^{k+1}, k = 0,1,\ldots
\end{align*}
which can be solved by a dominant eigenvector computation. We have empirically checked that this process converges to the Nash equilibrium of rock-paper-scissors starting from an interior initialization, but it may fail to converge in some randomly generated payoff matrices. We leave it as a question for future work how to apply algorithmic boosting to this iterative process.

\section{Future work}
\label{Conclusion}

Let us further single out a pair of questions for the future: (1) The first is to develop techniques for computing Nash equilibria (whether in zero-sum games or in the general case) using Pad\'e approximation theory: We have previously discussed how Pad\'e approximation theory can be used to compute the matrix exponential. This suggests a general method to compute a Nash equilibrium, namely, to prove that an algorithmically powered version of a dynamic converges to a Nash equilibrium and then to approximate that Nash equilibrium using Pad\'e approximation theory. Such an idea should render a polynomial-time algorithm to compute a Nash equilibrium under an arbitrary payoff matrix. (2) The second is to develop techniques for solving the system of equations
\begin{align*}
\dot{X}(i) &=  X(i) ((CZ)_i - X \cdot CZ), \quad i=1, \ldots, n\\
X &= \frac{1}{t} \int_{0}^t Z_\alpha d\alpha.
\end{align*}
We conjecture that this system has a solution that converges to a symmetric Nash equilibrium of $(C, C^T)$. In fact, our clairvoyant averaging algorithm is one possible discretization of this equation.\\

\section*{Acknowledgments}

The idea that Hedge is a discretization of the replicator dynamic is by Professor Yannis Kevrekidis. The first author would like to thank him for this contribution and other helpful discussions. We also thank Professor Avi Wigderson for a helpful discussion and for bringing some references to our attention.

\bibliographystyle{abbrvnat}
\bibliography{real}

\end{document}